\documentclass[10pt,english,oneside,a4paper]{amsart}

\usepackage[utf8x]{inputenc}
\usepackage[T1]{fontenc}
\usepackage{varioref}
\usepackage{amsmath}
\usepackage{amssymb}

\usepackage{concmath}

\makeatletter
\theoremstyle{plain}
\theoremstyle{plain}
\newtheorem{theorem}{Theorem}
  \theoremstyle{plain}
  \newtheorem{lemma}[theorem]{Lemma}
  \theoremstyle{plain}
  \newtheorem{definition}[theorem]{Definition}
  \theoremstyle{plain}
  \newtheorem{proposition}[theorem]{Proposition}
  \theoremstyle{remark}
  \newtheorem*{note*}{Note}
  \theoremstyle{remark}
  \newtheorem*{conclusion*}{Conclusion}
  \theoremstyle{remark}
  \newtheorem{note}[theorem]{Note}
 \theoremstyle{definition}
  \newtheorem{example}[theorem]{Example}
  \theoremstyle{plain}
  \newtheorem{corollary}[theorem]{Corollary}

\usepackage{amsfonts}
\usepackage{amscd}
\usepackage[cmtip,arrow]{xy}
\usepackage{pb-diagram,pb-xy}
\usepackage{hyperref}

\newcommand{\g}{\mathfrak{g}}

\newcommand{\id}{\mbox{id}}

\newcommand{\Ad}{\text{Ad}}
\newcommand{\ad}[1]{\text{ad}(#1)}

\newcommand{\dif}{\textrm{\textbf{d}}}

\newcommand{\Sym}[1]{\mathop{\mathsf{Symm}}{\left(#1\right)}}

\newcommand{\cI}{{\mathcal I}}
\newcommand{\cJ}{{\mathcal J}}
\newcommand{\cL}{{\mathcal L}}

\newcommand{\mR}{\mathbb{R}}

\newcommand{\be}{\textbf{e}}

\newcommand{\so}[1]{\mathfrak{so}(#1)}

\newcommand{\gl}{\mathfrak{gl}}

\newcommand{\tr}[1]{\mathsf{tr}\,{(#1)}}
\newcommand{\Tr}[1]{\mathsf{Tr}\,{(#1)}}

\makeatother

\usepackage{babel}

\begin{document}

\title{Differential geometry, Palatini gravity and reduction}

\author{S. Capriotti}
\address{Departamento de Matem\'atica, \mbox{Universidad Nacional del Sur}, 8000 Bah\'{\i}a Blanca, Argentina.}

\begin{abstract}
  The present article deals with a formulation of the so called \emph{(vacuum) Palatini gravity} as a general variational principle. In order to accomplish this goal, some geometrical tools related to the geometry of the bundle of connections of the frame bundle $LM$ are used. A generalization of Lagrange-Poincaré reduction scheme to these types of variational problems allows us to relate it with the Einstein-Hilbert variational problem. Relations with some other variational problems for gravity found in the literature are discussed.
\end{abstract}



\email{santiago.capriotti@uns.edu.ar}

\thanks{
  The author thanks CONICET for finantial support through a posdoctoral grant, and as a member of research projects PIP 11220090101018 and PICT 2010-2746 of the ANPCyT. This work is part of the IRSES project GEOMECH (nr. 246981) within the 7th European Community Framework Programme. Last, but not least, he would like to offer special thanks to the referees, whose remarks and suggestions improve substancially the quality and understandability of the paper. 
}

\maketitle

\setcounter{tocdepth}{2}

\tableofcontents

\section{Introduction}
In dealing with quantization and symmetry aspects of physical theories, it is important to have at our disposal a variational formulation; this is true even from the mathematical viewpoint, where methods for existence of solutions of PDEs are based in a variational version of these equations \cite{book:7463,dacorogna2004introduction}. In a pioneering work, Griffiths \cite{book:852048} (based on ideas of Cartan \cite{opac-b1120127}) extends the notion of variational problem by allowing the family of curves to be varied to live in the set of integral curves of a differentially closed ideal $\cI$ in the exterior algebra of $F$, an example of the so called \emph{exterior differential systems} (for definitions and standard results, we refer to \cite{BCG}.) These ideas (usually referred to as \emph{Griffiths formalism}) were further developed in \cite{hsu92:_calcul_variat_griff,MR924805}; an application of this formalism to the geometric control of quantum systems can be found in \cite{moseley2004geometric}. Gotay in \cite{GotayCartan} uses this generalized notion of variational problem as a mean to deal with the problem of properly define Cartan forms for field theories.
\\
The purpose of the following article is to find a particular formulation of the so called ``Palatini gravity'' (in fact, metric-affine gravity, see \cite{Hehl19951} and below) as a variational problem in the generalized sense described above. The precise meaning of such thing is established in Definition \ref{Def:VariationalProblems}, which can be found in Appendix \ref{Sec:VarProb+Dirac}; in short, it means to find a bundle $F\rightarrow M$ on the spacetime $M$, a form $\lambda\in\Omega^n\left(F\right)$ (here $n$ is the dimension of $M$) and a set of restrictions for the admissible sections of the bundle $F$, encoded as an exterior differential system (EDS from now on.) Concretely, we will obtain the vacuum Einstein equations as extremals of the functional
\[
\sigma\mapsto\int_M\sigma^*\lambda
\]
where $\sigma$ is a section of the bundle $F$ and integral for an EDS $\cI$, i.e., $\sigma^*\alpha=0$ for every $\alpha\in\cI$. This last requeriment imposes some restrictions on the allowed variations, because they must be tangent to the ``submanifold'' of integral sections for $\cI$; in this regard it is similar to \emph{vakonomics mechanics} as defined in \cite{cortes2002geometric}. They are different from other kind of constraints in the fact that induce restrictions not only in the shape of the extremals, but also on the allowed variations. The usual variational problems considered in classical field theory are a particular instance of the kind of variational problems we have in mind; in fact, in this case the restrictions are those forcing the degrees of freedom along the fibers of the map $J^1E\rightarrow E$ to be nothing but the derivatives of the degrees of freedom associated to the fibers of the bundle $E\rightarrow M$ (see Appendix \ref{Sec:VarProb+Dirac}, Section \ref{SubSubSec:FieldTheoryAsVarProblem}, where the variational problem for first order field theories are found to be particular examples of this definition.) These kind of variational problems, whose total bundle is a jet space and whose set of restrictions coincides with the \emph{contact structure} of the jet bundle, will be called \emph{classical variational problems}, to properly distinguish them from the version of variational problems we will work here. Thus the concept of variational problem adopted in this work includes not only the usual field theory, but instances in which the relation between fields are of different nature: In fact, our description of Palatini gravity will use this additional freedom in order to properly encode the requeriments of metricity and lack of torsion (the actual definition can be found in Eq. \eqref{Eq:PGEDS} below.) From now on, whenever the term ``variational problem'' is used, it will refer to the mathematical problem singled out in Definition \ref{Def:VariationalProblems}, requiring in an implicit way the existence of an associated \emph{variational triple}.

Usually, the variational problems considered in GR are classical (albeit singular), and the following features can be singled out:
\begin{enumerate}
\item The underlying bundle is never explicitly mentioned (see below).
\item It is customary to provide just a local version for the Lagrangian form, with no hint about its globalization.
\item The implicit assumption is that the underlying bundle is the jet space of some other bundle, so the restriction ideal $\cI$ is the contact ideal of the jet space.
\end{enumerate}
Although these facts does not prevent people to work succesfully with classical variational problems in general, our point of view is that some advantages could be extracted from this complementary formulation; the reduction theory developed below could serve as an indication in this sense (see Section \ref{Sec:ReductionTheory}.) Another advantage of the approach taken in this work is that our description is not restricted to four-dimensional spacetimes.

Let us see these features in some examples borrowed from the literature; in order to set common grounds, let us describe some formulations for gravity we encountered, with emphasis in the geometrical structures lying below them:
\begin{itemize}
\item\textbf{Einstein-Hilbert formulation.} In this case we have a particular kind of classical variational problem (see for example \cite{wald2010general}) where the components $g_{\mu\nu}$ of the local expression of a metric are the degrees of freedom of the theory, the underlying bundle would be $J^2\Sigma$ (here $\Sigma\rightarrow M$ indicates the quotient bundle $\Sigma:=LM/SO(1,n-1)$, the bundle of metrics with the given signature \cite{KN1}), and the Lagrangian form is given by the scalar curvature associated to the Levi-Civita connection times the canonical volume.
\item\textbf{Einstein-Hilbert formulation with vielbeins.} There exists another kind of classical variational problem, described in \cite{Peldan:1993hi} where the degrees of freedom are given by the components of a tetrad field $e^\alpha_I$. The underlying bundle could be set as $J^2LM\rightarrow M$, where $LM\rightarrow M$ is the bundle of frames on $M$ and the Lagrangian form can be written by using the (local) curvature $2$-form of the canonical connection (uniquely determined from the first structure equation \cite{MR532831}, after assuming that it takes values in the Lie algebra $\so{1,3}$ of the Lorentz group) and a pair of elements of the dual coframe. These local forms allow us to define a global form once we realize that this Lagrangian form is invariant by the action of the Lorentz group $SO(1,3)$; this action comes from the lift of the canonical action of $GL\left(4\right)$ on $LM$.
\item\textbf{Einstein-Palatini formulation (or metric-affine formulation).} This formulation can be find in the groundbreaking work of Arnowitt, Deser and Misner \cite{citeulike:820116}; the degrees of freedom are the components of a metric $g$ and the components of an affine connection, so the underlying bundle could be set to $J^1\left(\Sigma\times_MC\left(LM\right)\right)$, where $C\left(LM\right)\rightarrow M$ is the bundle of connections on $M$. The Lagrangian form is the product of the trace (respect to the metric $g$) of the Ricci tensor associated to the connection, times the invariant volume of $g$.
\item\textbf{Einstein-Palatini formulation with vielbeins.} A variational problem for the Einstein-Palatini gravity within the vierbein formalism appears in \cite{PhysRevD.14.2505}; its degrees of freedom are the components of a local basis for the tangent bundle to the spacetime $M$, plus an affine connection on $M$; the underlying bundle appears to be in this case $J^1\left(LM\times_MC\left(LM\right)\right)$.
\\
This variational problem is also described in \cite{Peldan:1993hi}. A discusion of the geometry behind this example could be found in \cite{sardanashvily2002classical}; it is equivalent to the operation done in the previous type of variational problem, where the bundle $\Sigma\rightarrow M$ is changed by the bundle $LM\rightarrow M$ everywhere.
\\
The article \cite{2012GReGr..44.2337D} compares its version of Einstein-Palatini gravity with another description in terms of vielbeins (i.e., a basis on the tangent bundle of $M$, where $M$ has arbitrary dimension); these objects can be seen as local sections of the bundle $LM\rightarrow M$, and the local Lagrangian form is the trace of the curvature associated to the connection respect to the frame (times an invariant volume on $M$.) The only way to globalize such form is to restrict the structure group of $LM$ to the Lorentz group, which is equivalent to admit a section of the quotient bundle $\Sigma=LM/SO(1,n-1)$; they explicitly assumed this reduction of the structure bundle.
\\
The work \cite{Friedric-1978} uses the same degree of freedom as in \cite{2012GReGr..44.2337D}; the main differences between these references lie in the family of tensors used in order to parametrize the arbitrary connection.
\end{itemize}
Now, we must be cautious about the names we used above for the diverse variational problems, because there is no agreement in the literature about them; for example, sometimes Palatini gravity is gravity with vielbeins (as in \cite{baez94:_gauge}, where skew symmetry of the connection forms is assumed, meaning that some kind of relation is allowed between the metric and the connections.) In the present article we will adopt a mixed approach in our choice of nomenclature when referring to variational problems related to General Relativity: In general we will use the term ``Palatini gravity'' to refer to those variational problems where vielbeins are employed as degrees of freedom, but the denominations used in the previous list will be adopted in those cases where a most detailed language would be necessary (this remark will become specially true in the discussions we will carry out throughout Section \ref{Sec:Gravitywithvielbeins}.) On the other side, it must be stressed that the geometrical structures mentioned in the previous descriptions are not present in the original works, and are suggested here as an appropriate geometrical setting for every variational problem, suitable for comparisons with our own variational problems. Nevertheless, it is interesting to note a couple of facts:
\begin{itemize}
\item In every case where a vielbein or tetrad is used as a degree of freedom, the structure group of the relevant bundles (which is also called \emph{internal group}) must be reduced to the Lorentz group. This will be discussed in Subsection \ref{SubSec:GeometryPalatini}; in short, this fact could be tied to the passive transformation property assumed for the local connections forms, inherited from classical differential geometry.
\item The underlying bundle is always a jet bundle, and moreover, it is assumed that the variations keep invariant the contact structures of these jet spaces, and this fact is given in local coordinates as the well-known mantra ``variation commutes with partial derivative''.
\end{itemize}
We have enough background to discuss the variational problem we will propose in the present work: We want to find variational problems analogous to gravity with vielbeins, both in the Einstein-Palatini and Einstein-Hilbert cases, but along the lines of Definition \ref{Def:VariationalProblems}, which requires the search of a bundle whose sections corresponds to degrees of freedom of vacuum gravity, and whose dynamics could be determined from a global Lagrangian form $\lambda_{PG}$ together the restrictions coming from an ideal $\cI_{PG}$ in the exterior algebra of the bundle. If $LM\rightarrow M$ is the bundle of frames on $M$, the bundle we are looking for is $J^1LM\rightarrow M$ (an explanation of this choice will be provided below, see Subsection \ref{SubSec:DegreeAndRestrictions}), and the equations of motion result from the search of sections $\sigma$ which are integral for the differentially closed ideal $\cI_{PG}\subset\Omega^\bullet\left(J^1LM\right)$ and extremals of the functional
\[
\sigma\mapsto\int_M\sigma^*\lambda_{PG},
\]
where $LM\rightarrow M$ is the frame bundle (with structure group $GL\left(n\right)$), the $n$-form $\lambda_{PG}$ will be a global form on $J^1LM$, and the restrictions imposed by the ideal $\cI_{PG}$ are weaker to those imposed by the contact structure of $J^1LM$. In a nutshell, it is because the holonomic sections of $J^1LM$ are associated to flat connections, a too strong condition to be imposed on a gravitational field; thus, in the proposed variational problem, ``variation will not commute with derivations''. Because of the decomposition
\[
J^1LM=LM\times_MC\left(LM\right)
\]
this variational problem is a kind of metric-affine theory, because its degrees of freedom are vielbein and connections, with the metric being reconstructed from the vielbein via the usual formula
\[
g_{\mu\nu}:=\eta_{ij}e^i_\mu e^j_\nu;
\]
here $\eta\in\gl\left(n\right)$ is a fixed operator with signature $\left(n-1,1,0\right)$. The main differences with the approaches described above are the restrictions $\cI_{PG}$ (is not a contact structure!), the structure group, which is $GL\left(n\right)$ in our case, and the globalization of the Lagrangian form is achieved without reducing the internal group. It is interesting that this variational problem will be versatile enough to be related with other types of variational problems for gravity, as the \emph{$GL\left(4\right)$-invariant gravity} \cite{0264-9381-7-10-011}, which is similar to the $GL(4)$-gravity of Komar \cite{PhysRevD.30.305}.
\\
A crucial feature of our variational problem is that the underlying bundle is a principal bundle with structure group $GL\left(n\right)$; thus the variational problem inherites a canonical action by the general linear group. Nevertheless, neither the Lagrangian form nor the ideal $\cI_{PG}$ are invariant for the canonical $GL\left(n\right)$-action, but they are invariant by this action when restricted to the Lorentz group $SO\left(1,n-1\right)\subset GL\left(n\right)$ determined by the matrix $\eta$. This fact raises the question about if it is possible to reduce the variational problem for this action; it led us to consider a reduction scheme analogous to the well-known Euler-Poincaré \cite{lopez00:_reduc_princ_fiber_bundl} and Lagrange-Poincaré \cite{2003CMaPh.236..223L} reduction schemes for field theories. A proposal for a generalized reduction scheme is described, and it is applied to our variational problem in order to relate it with a variational problem equivalent to Einstein-Hilbert.

The structure of the article is as follows: In Section $2$ a brief discussion of a classical variational problem for Palatini gravity is given, and after introducing the necessary geometrical tools, a variational problem in the Griffiths sense is defined. Afterwards the dynamical problem settled by these data is discussed: The Euler-Lagrange equations are obtained, and a treatment of the metric underlying a solution of these equations is performed.  As a bonus, the language developed in this section allows a discussion of the relation between our variational problem and other types of variational problems found in the literature. Section $3$ is devoted to found a proper generalization for reduction of variational problems in the Griffiths sense. In order to reach this goal, it is necessary to discuss a notion of reduction for EDS found in the literature; once the desired generalization is formulated, it should be checked that it reduces to the Euler-Poincaré reduction for field theories when dealing with classical variational problems. Finally, this reduction scheme is used to quotient out the extra degrees of freedom introduced when working with gravity with vielbeins.


\section{Gravity with vielbeins}\label{Sec:Gravitywithvielbeins}

\subsection{The geometry in the classical approach to gravity with vielbeins}\label{SubSec:GeometryPalatini}

It is perhaps necessary to review the geometrical contents of the usual representation of gravity with vielbeins. It is based in the local description of a connection on a manifold $M$ by means of the so called \emph{moving frames} \cite{MR532831}. In this setting, $M$ is supplied with a collection $\left\{e_i,U\right\}$, where $\left\{U\right\}$ is an open covering of $M$ and $\left\{e_i\right\}$ is a basis of vector fields for $TU$; the index $i$ thus runs from $1$ to $n$, the dimension of the manifold $M$. A \emph{connection} in these grounds is a collection $\left\{\omega_i^j\right\}\subset\Omega^1\left(U\right)$ of $1$-forms on every $U$ of the covering, such that the following compatibility condition is met: Whenever $U\cap\overline{U}\not=\emptyset$, the two collections $\left(e_k,\omega^i_j\right)$ and $\left(\bar{e}_k,\bar{\omega}^i_j\right)$ define a map $M:U\cap\overline{U}\rightarrow GL\left(n\right)$ (the so called \emph{transition functions}) such that
\[
\bar{e}_k=M_k^le_l,
\]
and the local connection $1$-forms must be related by
\[
\bar{\omega}^i_j=N^i_k\dif M^k_j+N^i_k\omega^k_lM^l_j,
\]
where $N=M^{-1}$. The set $\left\{\omega_i^j\right\}$ of $1$-forms can be considered as the components of a local $\gl\left(n\right)$-valued $1$-form $\omega_U$, and they provide the local description for the covariant operator according to the formula
\[
\nabla e_i:=\omega^j_i e_j.
\]
It must be stressed that the introduction of the local basis $\left\{e_i\right\}$ made in the previous description is totally incidental. Nevertheless, this situation changes when we try to introduce an action for these degrees of freedom by means of the local formula\footnote{There exists a slight abuse of notation in this formula, namely the integrand is only defined in $U$, although the integral runs all over $M$.}
\begin{equation}\label{Eq:ClassicalPalatini}
  S\left[e,\omega\right]:=\int_M\epsilon_{ijkl}\eta^{ip}\theta^k\wedge\theta^l\wedge\left(\dif\omega_p^j+\omega^j_q\wedge\omega^q_p\right),
\end{equation}
where $\left\{\theta^k\right\}$ is the dual basis to $\left\{e_j\right\}$; on an overlap $U\cap\overline{U}$ this action will be uniquely determined if and only if $\eta_{ij}M^i_kM^j_l=\eta_{kl}$, i.e. when $M$ has values in the Lorentz group. But this would restricts the local basis to be a local orthonormal basis for a metric $g$, which is recovered from the formula $g=\eta_{ij}\theta^i\otimes\theta^j$. In order to avoid this hasty conclusion, it is necessary to have at our disposal a language where the different kind of transformations underlying this model would be apparent; our viewpoint is to work in the realm of fibre bundles, in particular, by using the bundles $LM\rightarrow M$ and $C\left(LM\right)\rightarrow M$ in the descriptions of gravity we want to work with. From this perspective, we could replace the forms in Eq. \eqref{Eq:ClassicalPalatini} by its global counterparts living on $J^1LM$; as arise from the local expressions, see Appendix \ref{subsubsect:LocalExpressions}, the apparent lack of invariance of the integrand in this equation could be explained from the fact that when it is defined in terms of local forms, it is actually the local $n$-form on $M$ obtained from this global $n$-form on $J^1LM$ by pulling back along a local section. The bundle $LM$ is a principal bundle with structure group $GL\left(n\right)$, and thus there exists a canonical action of this group on the bundles we work with; we will use this action in order to use a generalized idea of reduction as a way of quotient out additional degrees of freedom introduced in the theory by the frames. In terms of the so called \emph{active and passive transformations} (in the sense indicated by \cite[p.\ 406]{book:8133}) we are proposing a change of the passive transformations regarding the objects $e_i$ and $\omega^i_j$, although at the same time preserving the active transformations of the gravity with vielbeins; it is our understanding that the change of passive transformation is an operation analogous to the change of bundle from the mathematician's viewpoint, and the preservation of the active transformations could be achieved by working only with bundles associated to $LM$.

\subsection{Canonical forms on $J^1LM$}\label{SubSubsect:CanonicalForms}

It is necessary to introduce some canonical objects living in the exterior algebra of $J^1LM$, because these objects would be used in the search of the Lagrangian form and the right restrictions for our variational problem. The discussion carried out in Appendix \ref{SubSection:GeometryJP} allows us to define a $1$-form $\omega$, which is a canonical $\mathfrak{gl}\left(n\right)$-valued pseudotensorial $1$-form of type $\left(GL\left(n\right),\text{ad}\right)$ defined on $J^1LM$. As we know, on $LM$ there exists another canonical form, namely, the \emph{tautological form} $\bar\theta$ \cite{KN1}; thus the projection $\tau_{10}:J^1LM\rightarrow LM$ can be used in order to define a new canonical form on $J^1LM$, that is
\[
\theta:=\tau_{10}^*\bar\theta\in\Omega^1\left(J^1LM,\mR^n\right).
\]
Let us recall that under the identification $J^1LM\simeq p^*LM$ the canonical right action translates into
\[
\left(\rho,u\right)\cdot h=\left(\rho,u\cdot h\right).
\]
Using this fact, we can see that the form $\theta$ has the following remarkable properties.
\begin{proposition}
  The form $\theta$ is a tensorial $1$-form of type $\left(GL\left(n\right),\mR^n\right)$. Moreover, for every connection $\Gamma$ on $LM$, we have that
  \[
  \tilde\sigma_\Gamma^*\theta=\bar\theta.
  \]
\end{proposition}
\begin{proof}
  The second assumption follows easily from the definition of $\tilde\sigma_\Gamma$. Now let $h$ be an element of $GL\left(n\right)$; every element $Z\in T_{\left(\rho,u\right)}J^1LM$ can be written as $Z=\left(V,X\right)$, where $\tau_1\left(\rho\right)=\tau\left(u\right)$ and moreover
  \[
  T_\rho\tau_1\left(V\right)=T_u\tau\left(X\right).
  \]
  Thus we have that
  \begin{align*}
    \left.\left(R_{h}^*\theta\right)\right|_{\left(\rho,u\right)}\left(Z\right)&=\left.\theta\right|_{\left(\rho,u\cdot h\right)}\left(T_{\left(\rho,u\right)}R_h\left(V,X\right)\right)\\
    &=\left.\theta\right|_{\left(\rho,u\cdot h\right)}\left(\left(V,T_uR_hX\right)\right)\\
    &=\left.\bar\theta\right|_{u\cdot h}\left(T_uR_hX\right)\\
    &=h^{-1}\cdot\left(\left.\bar\theta\right|_u\left(X\right)\right)\\
    &=h^{-1}\cdot\left(\left.\theta\right|_{\left(\rho,u\right)}\left(V,X\right)\right)
  \end{align*}
  and $\theta$ is pseudotensorial of type $\left(GL\left(n\right),\mR^n\right)$. Finally, to show that $\theta$ is tensorial, we need to prove that $\theta\left(Z\right)=0$ if $Z$ is a vertical vector in $J^1LM\rightarrow J^1LM/GL\left(n\right)$. But
  \[
  V_{\left(\rho,u\right)}\left(J^1LM\right)=\left\{\left(V,0\right)\in T_\rho\left(J^1LM/GL\left(n\right)\right)\times T_uP:T_\rho\tau_1\left(V\right)=0\right\},
  \]
  and so, in particular, $\theta\left(Z\right)=0$ for every $Z\in VJ^1LM$.
\end{proof}

Then if $T\in\Omega^2\left(J^1LM,\mR^n\right)$ is the exterior covariant differential of $\theta$ obtained using the connection $\omega$, we will have that
\[
T=\dif\theta+\omega\stackrel{\wedge}{\cdot}\theta.
\]
Let $\Gamma$ be a connection on $LM$; by pulling back this expression along $\tilde\sigma_\Gamma$, we have that
\[
\tilde\sigma_\Gamma^*T=\dif{\bar\theta}+\omega_\Gamma\stackrel{\wedge}{\cdot}\bar\theta
\]
which is in turn equal to the torsion $T_\Gamma$ of the connection $\Gamma$. Thus $T$ can be called \emph{universal torsion}.

\subsection{Degrees of freedom and restriction EDS for gravity with vielbeins}\label{SubSec:DegreeAndRestrictions}

A \emph{tetrad} or, more generally, a \emph{vielbein}, is a local isomorphism
\[
e:TM\rightarrow M\times\mR^n
\]
or equivalently, a basis for the tangent bundle to $M$ on an open set $U\subset M$. The rationale behind these objects is simply to replace the basis of the tangent bundle induced by the coordinates by a more general basis, perhaps determined by geometrical insights related to the formulation of the problem at hands, in order to simplify some equations living in a tensor bundle. In fact, our approach to gravity is based in the replacing of the metric by an (by definition) orthonormal local frame; a change in the metric is thus performed by a change in the vielbein. Nevertheless, it is necessary to point out an essential difference between our approach and some of the descriptions of the gravity with vielbeins that can be found in the literature (see for example \cite{baez94:_gauge}), where it is assumed the tangent bundle $TM$ to be isomorphic to $M\times\mR^n$ (the ``fake tangent bundle'' viewpoint): We keep the local character of the frame bundle, using it only as a tool that permits us to describe a connection on $M$, without additional topological assumptions on this manifold, namely, without considering it as a parallelizable manifold.

From the mathematical viewpoint, a vielbein is nothing but a local section of the frame bundle $LM$; a connection on $M$, on the other side, can be considered as a section of the bundle of connections $C\left(LM\right)$. We will take then as fields the sections of the bundle $p_1:J^1LM\rightarrow M$. It could be motivated by the fact that we have to describe a metric through a local basis (namely, a section of $LM$) and an independent connection, which is a section of $C\left(LM\right)$; the degrees of freedom will be sections of the product bundle
\[
LM\times_MC\left(LM\right)\rightarrow M
\]
which is isomorphic to $J^1LM$ as a bundle on $C\left(LM\right)$ (see Proposition \ref{Prop:IsoBetweenPpalBundles}.) Nevertheless, it will be necessary to impose some restrictions on these connections, because it must be related to the metric implicitly described by the vielbein (we are using here the notation developed at Subsection \ref{Subsect:UsefulNotation}, see below):
\begin{itemize}
\item \textbf{Metricity: }If the connection provides us with a covariant derivative, it will be desirable the preservation of the implicit metric, and a local section $s:M\rightarrow J^1LM$ will fulfill this requeriment if and only if
  \[
  s^*\left(\eta^{ik}\omega_k^j+\eta^{jk}\omega_k^i\right)=0,
  \]
  where $\omega\in\Omega^1\left(J^1LM,\gl\left(n\right)\right)$ is the canonical connection form on $J^1LM$ and $\eta$ is a Lorentzian fixed metric on $\mR^n$ ($n=\dim{M}$, of course.) 
\item \textbf{Zero torsion: } Additionally it will be required that
  \[
  s^*T=0
  \]
  for $T\in\Omega^1\left(J^1LM,\mR^n\right)$ the canonical form corresponding to the torsion.
\end{itemize}
Similar conditions can be found in the literature; see for example \cite{Eguchi1980213}. They differ, however, in a fundamental fact with our restrictions: While the conditions found in the literature (as far as we know) are of local character, ours are global. It is tied with the spaces where they live: In the former case they are forms on the spacetime, whereas in the latter the conditions are forms on the jet space of $LM$. The choice of the notation made here could led to confusion, because of the similarities with the notation used elsewhere to indicate local forms; it is important to point out these differences, in order to avoid potential confusions.
\\
These restrictions will be adopted to be the restriction EDS for Einstein-Hilbert gravity with vielbeins. Thus, although the underlying bundle for this variational problem is a jet space, the restriction EDS will be different from the contact structure; it is not totally unexpected, because the contact structure imposes on a local section the requeriment
\[
s^*\omega=0,
\]
forcing the connection to be flat, a too strong condition for a vacuum gravitational field.

\subsection{The Lagrangian form of gravity with vielbeins}

Let us make use of the geometrical constructions detailed in Appendix \ref{SubSection:GeometryJP} and Subsection \ref{SubSubsect:CanonicalForms} in order to find a useful description of the Lagrangian form for gravity with vielbeins (in both cases, Einstein-Hilbert and Einstein-Palatini); the notation is described there. In order to formulate a Lagrangian on $J^1LM$ we recall that there exists on this space a $\mR^n$-valued $1$-form $\theta$, namely, the pullback along $\tau_{10}$ of the canonical form $\bar\theta$ on $LM$, and, additionaly, the $\mathfrak{gl}\left(n\right)$-valued $2$-form $\Omega$, just constructed as the curvature $2$-form associated to the canonical connection on $J^1LM$ induced by the contact structure. For every $k=1,\cdots,n$ we can define the $\bigwedge^k\left(\mR^n\right)$-valued $k$-form
\[
\theta^{\left(k\right)}\left(X_1,\cdots,X_k\right):=\theta\left(X_1\right)\wedge\cdots\wedge\theta\left(X_k\right),\quad X_1,\cdots,X_k\in\mathfrak{X}\left(J^1LM\right)
\]
and it allows us to define the $\bigwedge^2\left(\mR^n\right)$-valued $n-2$-form $\Xi$ via
\[
\tilde\Xi:=\ast\left(\theta^{\left(n-2\right)}\right),
\]
where $\ast:\bigwedge^{n-i}\left(\mR^n\right)\rightarrow\bigwedge^i\left(\mR^n\right)$ is the Hodge star operator induced on the exterior algebra of $\mR^n$ by $\eta$. Therefore we can use the antisymmetrization operator 
\[
A:\bigwedge^2\left(\mR^n\right)\rightarrow\mR^n\otimes\left(\mR^n\right)^*=\left(\mathfrak{gl}\left(n\right)\right)^*:u\wedge v\mapsto \frac{1}{2}\left[v\otimes\eta\left(u,\cdot\right)-u\otimes\eta\left(v,\cdot\right)\right]
\]
in order to define a $\left(\mathfrak{gl}\left(n\right)\right)^*$-valued $n-2$-form, namely
\[
\Xi:=A\left(\tilde\Xi\right);
\]
the \emph{Palatini Lagrangian}\footnote{We will use this name by the Lagrangian, because it serves in both cases, Einstein-Hilbert and Einstein-Palatini, reserving the name Einstein Lagrangian to name the Lagrangian written in terms of metric and connection coefficients.} is
\begin{equation}\label{Eq:PalatiniLagrangian}
\lambda_{PG}:=\left<\Xi\stackrel{\wedge}{,}\Omega\right>
\end{equation}
where $\left<\cdot\stackrel{\wedge}{,}\cdot\right>:\Omega^k\left(J^1LM,\left(\mathfrak{gl}\left(n\right)\right)^*\right)\otimes\Omega^l\left(J^1LM,\mathfrak{gl}\left(n\right)\right)\rightarrow\Omega^{k+l}\left(J^1LM\right)$ indicates the extension of the contraction between elements of a vector space and its dual to the wedge product of $\mathfrak{gl}\left(n\right)$ and $\left(\mathfrak{gl}\left(n\right)\right)^*$-valued forms.

\begin{note}
{The extension of operations from $\g$ and $\mR^n$ to $\g$- and $\mR^n$-valued forms is detailed in the work of Kôlar \emph{et al.} \cite{Kolar93naturaloperations}, p. 100. In particular, the product structure on the $\mR^n$-valued forms that yields to the form $\theta^{\left(n-2\right)}$ can be considered as an specialization of a more general structure found on vector valued forms.}
\end{note}


\subsection{The Lagrangian form in terms of a basis on $\mR^n$}\label{Subsect:UsefulNotation}

We can introduce a basis on $\mR^n$; this allows us to write out the Lagrangian in terms of components of the canonical forms, giving some expressions that will be useful in handling with the variations. Namely, if $\left\{\be_1,\cdots,\be_n\right\}$ is the canonical basis on $\mR^n$, we can write $\theta:=\theta^i\be_i$ for some collection $\left\{\theta^i\right\}$ of $1$-forms, and from here it can be concluded that
\[
A\left(\theta^{\left(2\right)}\right)=\theta^i\wedge\theta^j\eta\left(\be_i,\cdot\right)\be_j.
\]
As for the operator $\ast$, we can conclude with the formula
\[
\ast\left(\theta^i\wedge\theta^j\right)=\eta^{ik}\eta^{jl}\theta_{kl},
\]
where is was introduced the set of $n-p$-forms
\begin{align*}
\theta_{i_1\cdots i_p}&:=\frac{1}{\left(n-p\right)!}\epsilon_{i_1\cdots i_pi_{p+1}\cdots i_n}\theta^{i_{p+1}}\wedge\cdots\wedge\theta^{i_n}\cr
&=X_{i_p}\lrcorner\cdots\lrcorner X_{i_1}\lrcorner\sigma_0;
\end{align*}
these forms are useful when dealing with the so called \emph{Sparling forms}, see \cite{MR875299}. Therefore for taking $\left\{\be^1,\cdots,\be^n\right\}$ as the dual basis of $\left\{\be_1,\cdots,\be_n\right\}$, 
\begin{align*}
\Xi&:=\eta^{ik}\eta^{jl}\theta_{kl}\eta\left(\be_i,\cdot\right)\be_j\cr
&=\eta^{jl}\theta_{kl}\be^k\otimes\be_j,
\end{align*}
and the Palatini Lagrangian can be written as
\begin{equation}\label{Eq:PalatiniLagrangianLocal}
\lambda_{PG}=\eta^{kp}\theta_{kl}\wedge\Omega^l_p.
\end{equation}
It must be stressed at this point that $\lambda_{PG}$ is a global form, despite the fact that the bilinear form $\eta$ is involved in its definition: This is an straightforward consequence of Eq. \eqref{Eq:PalatiniLagrangian}, but it can be seen directly from this expression, by taking into account that $\Omega^j_k$ are the components (in the basis of $\mR^n$ making $\eta=\left(\eta_{ij}\right)$, c.f. the basis adopted above) of the curvature of the canonical connection on the $GL\left(n\right)$-principal bundle $J^1LM\rightarrow C\left(LM\right)$, and $\theta^k$ are the components of the $\mR^n$-valued canonical form $\theta$ in the same basis (see Subsection \ref{SubSubsect:CanonicalForms}.)

\subsection{The structure equations}

There are some equations that we need to take into account in this work. First we have the \emph{structure equations}
\begin{align*}
  &\dif\omega^i_j+\omega^i_k\wedge\omega^k_j=\Omega^i_j\\
  &\dif\theta^i+\omega^i_k\wedge\theta^k=T^i,
\end{align*}
then its differential consequences, namely, the \emph{Bianchi identities}
\begin{align*}
  &\dif\Omega^i_j=\Omega^i_k\wedge\omega^k_j-\omega^i_k\wedge\Omega^k_j\\
  &\dif T^k=\Omega^k_l\wedge\theta^l-\omega^k_l\wedge T^l,
\end{align*}
and some additional related identities
\begin{align*}
  &\dif\theta_{li}=\omega^k_l\wedge\theta_{ki}-\omega^k_i\wedge\theta_{kl}-\omega^s_s\wedge\theta_{li}+T^k\wedge\theta_{lik}\\
  &\dif\Omega^{pq}=\Omega^p_k\wedge\omega^{kq}-\omega^p_k \wedge\Omega^{kq}\\
  &\dif\omega^{lp}=-\omega^l_s \wedge\omega^{sp}+\Omega^{lp}\\
  &\dif\theta_{ipq}=\omega^k_i \wedge\theta_{kpq}+\omega^k_p \wedge\theta_{kqi}+\omega^k_q \wedge\theta_{kip}-\omega_s^s \wedge\theta_{ipq}+T^k \wedge\theta_{ipqk}
\end{align*}
where it were introduced the handy notations $\omega^{ij}:=\eta^{jp}\omega^{i}_p,\Omega^{ij}:=\eta^{jp}\Omega^{i}_p$.

\subsection{Dynamics of Einstein-Hilbert gravity with vielbeins}

After this rather lenghty warm-up, we are ready to describe now the gravity with vielbeins as a variational problem in the sense adopted in this work; as the bundle of the variational triple for this theory we will take the bundle $J^1LM\rightarrow M$; the Lagrangian form on this space will be $\lambda_{PG}$, defined in Eq.\ \eqref{Eq:PalatiniLagrangian}. Finally, the EDS restricting properly the sections of $J^1LM$ is the one generated by metricity and torsionless conditions, namely
\begin{equation}\label{Eq:PGEDS}
  \cI_{PG}=\left<\pi_{\mathfrak{p}}\omega,T\right>_{\text{diff}};
\end{equation}
here $\pi_{\mathfrak{p}}:\gl\left(n\right)\rightarrow\mathfrak{p}$ is the projection onto the second summand in the Cartan decomposition $\gl\left(n\right)=\mathfrak{so}\left(1,n-1\right)\oplus\mathfrak{p}$ induced by $\eta$; in the basis introduced above, this projection reads
\[
\left(\pi_{\mathfrak{p}}\omega\right)^{ij}:=\eta^{ik}\omega_k^j+\eta^{jk}\omega_k^i,
\]
which is nothing but the metricity condition mentioned before. The variational problem we are proposing for (vacuum) GR with vielbeins is the variational problem associated to the triple
\[
\left(J^1LM\rightarrow M,\lambda_{PG},\cI_{PG}\right);
\]
a quick comparison with some classical variational problems found in literature (cf. those mentioned into the introductory section) tell us that it has fewer degrees of freedom than other alternatives.

Finally, it is necessary to point out that considerations of variational problems on the frame bundle, although from a slighty different point of view, can be found in the literature \cite{brajercic04:_variat}.





\subsubsection{Considerations about admissible variations}

Given the existence of a restriction EDS, the variations to be considered in order to find out the equations of motion of gravity with vielbeins cannot be arbitrary; rather they must be restricted in some way. Let us recall that a \emph{variation} of a section $s:M\rightarrow E$ of a bundle $E\rightarrow M$ is a section of the pullback bundle $s^*\left(VE\right)$, perhaps with compact support, and that the relevant variations for a variational problem are the infinitesimal symmetries of the restriction EDS. We could introduce the following definition in order to work here with these objects.

\begin{definition}[Admissible variations]\label{Def:AdmVars}
An \emph{admissible variation} of the integral section $s$ for an EDS $\cI$ is a variation $\delta s$ with an extension $\widehat{\delta s}\in\mathfrak{X}\left(E\right)$ which is an infinitesimal symmetry of $\cI$, that is, such that
\[
s^*\left(\cL_{\widehat{\delta s}}\cI\right)=0.
\]
\end{definition}
An admissible variation for an EDS $\cI$ produces a path in the set of integral sections of this EDS. 
In terms of adapted coordinates (see Subsection \ref{subsubsect:LocalExpressions} in Appendix \ref{SubSection:GeometryJP}) $\left(x^\mu,e^\nu_k,e^\sigma_{k\rho}\right)$ on $J^1LM$, any variation reads
\[
x^\mu\mapsto\left(0,\delta e^\nu_k,\delta e^\sigma_{k\rho}\right);
\]
it means in particular that the canonical forms $\theta$ and $\omega$ can be varied independently. This freedom will be use in order to simplify the calculations below.

\subsubsection{The variations of the connection}
On the same open set $U_\alpha$ and using the previous identifications, we can consider that variations of the connection as $\mathfrak{gl}\left(n\right)$-valued $1$-forms $\delta\omega^{ij}$.
Therefore
\begin{align}
\delta_\omega\lambda_{PG}&=\theta_{ik}\wedge\left[\dif\left(\delta\omega^{ik}\right)+\eta_{pq}\delta\omega^{pi}\wedge\omega^{kq}+\eta_{pq}\omega^{pi}\wedge\delta\omega^{kq}\right]\cr
&=\left(-1\right)^{n+1}\dif\theta_{ik}\wedge\delta\omega^{ik}+\theta_{ik}\wedge\left(-\eta_{pq}\omega^{kq}\wedge\delta\omega^{pi}+\eta_{pq}\omega^{pi}\wedge\delta\omega^{kq}\right)\cr
&=\left(-1\right)^{n+1}\left(\eta_{ip}\omega^{lp}\wedge\theta_{lk}-\eta_{kp}\omega^{lp}\wedge\theta_{li}+T^l\wedge\theta_{ikl}\right)\wedge\delta\omega^{ik}+\cr
&\qquad\qquad\qquad+\theta_{ik}\wedge\left(-\eta_{pq}\omega^{kq}\wedge\delta\omega^{pi}+\eta_{pq}\omega^{pi}\wedge\delta\omega^{kq}\right)\cr
&=\left[\left(-1\right)^{n+1}\left(\eta_{ip}\omega^{lp}\wedge\theta_{lk}-\eta_{kp}\omega^{lp}\wedge\theta_{li}+T^l\wedge\theta_{ikl}\right)-\eta_{iq}\theta_{kl}\wedge\omega^{lq}+\eta_{pk}\theta_{li}\wedge\omega^{pl}\right]\wedge\delta\omega^{ik}\cr
&=\left[-\eta_{ip}\theta_{lk}\wedge\omega^{lp}+\eta_{kp}\theta_{li}\wedge\omega^{lp}+\left(-1\right)^{n+1}T^l\wedge\theta_{ikl}-\eta_{iq}\theta_{kl}\wedge\omega^{lq}+\eta_{pk}\theta_{li}\wedge\omega^{pl}\right]\wedge\delta\omega^{ik}\cr
&=\left[\eta_{kp}\theta_{li}\wedge\left(\omega^{lp}+\omega^{pl}\right)+\left(-1\right)^{n+1}T^l\wedge\theta_{ikl}\right]\wedge\delta\omega^{ik}.\label{Eq:ConnectionVariations}
\end{align}
Therefore the variations of the Lagrangian $\lambda_{PG}$ annihilates, independently of the form of the variations $\delta\omega$, and so its does not contribute to the equations of motion.

\subsubsection{Considerations about the variations of the frame}

It is time to see what the variations of the frame produce on the $n-2$-forms $\theta_{ij}$ defined previously. We will consider here variations with its support on a chart $U_\alpha$. Now,
\begin{multline*}
\cL_{\delta s}\theta_{ij}=\left(-1\right)^{i+j+1}\Bigg[\delta\theta^1\wedge\cdots\wedge\widehat{\theta^i}\wedge\cdots\wedge\widehat{\theta^j}\wedge\cdots\wedge\theta^n+
\cdots+\theta^1\wedge\cdots\wedge\widehat{\theta^i}\wedge\cdots\wedge\widehat{\theta^j}\wedge\cdots\wedge\delta\theta^n\Bigg].
\end{multline*}
Therefore
\begin{align}
\cL_{\delta s}\theta_{ij}&=\delta\theta^k\wedge\theta_{kij}.\label{Eq:GeneralThetaVariation}
\end{align}
So by performing the variations of the frame, we obtain 
\begin{align*}
\delta_{\xi_1}\lambda_{PG}&=\delta\theta^m\wedge\theta_{mki}\wedge\left(\dif\omega^{ki}+\eta_{lm}\omega^{li}\wedge\omega^{km}\right),
\end{align*}
namely
\begin{equation}\label{Eq:EulerLagrangeFrame0}
\theta_{jki}\wedge\left(\dif\omega^{ki}+\eta_{lm}\omega^{li}\wedge\omega^{km}\right)=0.
\end{equation}

As an additional formula useful in dealing with the variations of the frame, we can calculate the differential of the forms $\theta_{ij}$, expressing them in terms of the connection and the associated torsion. Namely, by using the definition
\[
\dif\theta^i=-\omega^i_k\wedge\theta^k+T^i
\]
we will obtain that
\[
\dif\theta_{ij}=\omega^k_i\wedge\theta_{kj}-\omega^k_j\wedge\theta_{ki}-\omega^k_k\wedge\theta_{ij}+T^k\wedge\theta_{ijk}
\]
where, as above
\[
\theta_{ijk}:=X_k\lrcorner X_j\lrcorner X_i\lrcorner\sigma_0.
\]

As shown in \cite{MR838749}, the equations of motion \eqref{Eq:EulerLagrangeFrame0} are equivalent to the annihilation of the Einstein tensor.

\subsection{Discussion: The global form for Einstein equations and the underlying metric}

According to the previous calculations, the equations of motion for the gravity with vielbeins can be described as the EDS generated by the forms
\begin{equation}\label{Eq:GlobalEinsteinEqs}
  \begin{cases}
    \dif\theta^i+\omega^i_k\wedge\theta^k,&\cr
    \theta_{ipq}\wedge\Omega^{pq},&\cr
    \eta^{ip}\omega_{p}^j+\eta^{jp}\omega_{p}^i.
  \end{cases}
\end{equation}
It is interesting to note that these expressions are global; we can think on them as a global form for vacuum Einstein equation. Additionally the jet space $J^1LM$ has a $GL\left(n\right)$-action, obtained by lifting the corresponding action of $GL\left(n\right)$ to the frame bundle; in terms of the adapted coordinates, it reads
\[
g\cdot\left(x^\mu,e^\nu_k,e^\sigma_{k\rho}\right)=\left(x^\mu,g_k^le^\nu_l,g^l_ke^\sigma_{l\rho}\right).
\]
This action is involved in the proof of the next proposition, giving sense to our choice of the relevant fields for describing gravity.
\begin{proposition}
  Let $s,\bar{s}:U\rightarrow J^1LM$ be a pair of solutions for the problem posed by \eqref{Eq:GlobalEinsteinEqs} on the connected open set $U\subset M$. Then there exists an smooth map $g:U\rightarrow SO\left(1,n-1\right)$ such that $\bar{s}=g\cdot s$.
\end{proposition}
\begin{proof}
If $s:U\rightarrow J^1LM$ is a local solution for these equations and $k:U\rightarrow SO\left(1,n-1\right)$ is an smooth map, we will have that ${s}':=k\cdot s$ verifies
\[
{s}'^*\left(\dif\theta^i+\omega^i_k\wedge\theta^k\right)=0={s}'^*\left(\theta_{ipq}\wedge\Omega^{pq}\right)
\]
and
\begin{equation}
  \label{Eq:SkewSymmetricConnection}
  {s}'^*\left(\eta^{ip}\omega_{p}^j+\eta^{jp}\omega_{p}^i\right)=0; 
\end{equation}
so it remains to show that every change of basis can be reduced to a change of basis in the Lorentz group. Now for $s,\bar{s}$ there exists $h:U\rightarrow GL\left(n\right)$ such that $\bar{s}=g\cdot s$; by using a Cartan decomposition of $GL\left(n\right)$ respect to the form $\eta$ we can factorize $GL\left(n\right)=P\cdot SO\left(1,n-1\right)$ where $P$ is the set of $\eta$-symmetric matrices, and if $h=p\cdot k$ in this factorization, we see from Eq. \eqref{Eq:SkewSymmetricConnection} (with the replacement $s'\rightarrow\bar{s}$) that the $\eta$-symmetric factor $p$ must verify
\[
\Ad_{p^{-1}}\omega+p^{-1}\dif p\in\so{1,n-1}
\]
for $\omega\in\so{1,n-1}$ (for recalling that $s$ verifies an equation analogous to Eq. \eqref{Eq:SkewSymmetricConnection}.) But there exists $a:U\rightarrow SO\left(1,n-1\right)$ such that $p=aca^{-1}$, where $c:U\rightarrow P$ is a diagonal matrix; therefore the previous requeriment on $p$ translates into
\[
\Ad_{c^{-1}}\tilde\omega+c^{-1}\dif c\in\so{1,n-1}
\]
with $\tilde\omega\in\so{1,n-1}$. But the first summand in this expression has zero entries in the diagonal, and the second is a diagonal matrix, so it can be split as the pair of conditions
\[
\Ad_{c^{-1}}\tilde\omega\in\so{1,n-1},\qquad c^{-1}\dif c=0
\]
for $\tilde\omega\in\so{1,n-1}$; this means that $c$ must be locally constant, and the first forces $c=\text{Id}$.
\end{proof}
Therefore the local solutions determine a unique metric $h$ according to the formula
\[
h:=\eta_{kl}e^k_\mu e^l_\nu\dif x^\mu\otimes\dif x^\nu.
\]
The first and third forms in the above EDS are enough to determine uniquely the connection, or more precisely the functions $e^\mu_{k\nu}$ of a solution, from the frame functions $e^\mu_k$. This result follows at once by using Proposition \ref{Prop:UniqueSolution}; as we will see below, these functions determines a connection on $M$ that is the Levi-Civita connection for $h$. Therefore the second set of forms are the true equations of motion for the metric.

\begin{note}[$GL\left(4\right)$-invariant gravity]
  This setting provides us with enough tools in order to describe some other approaches to ``gravity with moving frames''. For example, we can set a variational problem for the so called \emph{$GL\left(4\right)$-invariant gravity} \cite{0264-9381-7-6-007,0264-9381-7-10-011}: In brief, in this theory the fields are a soldering form $\theta\in\mathop{\text{Iso}}{\left(TM,TM\right)}$, a metric $\kappa$ and a connection $\Gamma$ on $M$. Using the previous identifications, we can consider these fields as sections of the bundles $\left(LM\times_M LM\right)/GL\left(n\right)$ (where $GL\left(n\right)$ is acting diagonally), $\Sigma$ and $C\left(LM\right)$ respectively. By using that the metric $\kappa$ is induced locally by a section $e$ of $LM$, and that such section induces a local section
\[
\tilde{e}:\left(LM\times_M LM\right)/GL\left(n\right)\rightarrow LM\times_M LM
\]
via
\[
\left[f_1,f_2\right]_{GL\left(n\right)}\mapsto\left(e\left(x\right),g\cdot f_2\right)\quad\text{iff }x:=p\left(f_1\right)\text{ and }e\left(x\right)=g\cdot f_1,
\]
the bundle $LM\times_M LM$ can be used instead of the first two bundles mentioned above; it amounts to describe the morphism $\theta$ by the way it is acting on a particular basis of $TM$, namely, the basis used in the description of the metric $\kappa$. The underlying bundle of the variational problem describing this kind of gravity theory will be $C\left(LM\right)\times_M LM\times_M LM=J^1LM\times_MLM$, which can be considered as a submanifold of $J^1LM\times_M J^1LM$ via the inclusion
\[
\imath:C\left(LM\right)\times_M LM\times_M LM\hookrightarrow J^1LM\times_M J^1LM:\left(\Gamma,e,f\right)\mapsto\left(\Gamma,e;\Gamma,f\right).
\]
By denoting $p_A,A=1,2$ the projections onto the first and second factor in $J^1LM\times_MJ^1LM$, the restriction EDS is generated as follows
\[
\cI_{FP}:=\left<\imath^*p_1^*\left(\eta^{ik}\omega_k^j+\eta^{jk}\omega_k^i\right),\imath^*p_2^*T\right>_{\text{diff}};
\]
these restrictions are nothing but Eqs. (2.1) and (2.2) in \cite{0264-9381-7-6-007}. Finally the Lagrangian considered by these authors is the pullback along $p_2$ of the Palatini Lagrangian defined above $\lambda_{PG}$, namely
\[
\lambda_{FP}:=p_2^*\lambda_{PG}.
\]
Thus the variational problem for this version of gravity is the triple
\[
\left(C\left(LM\right)\times_MLM\times_MLM,\lambda_{FP},\cI_{FP}\right).
\]
\end{note}

\subsection{Differential consequences of the vacuum Einstein equations}

As an additional result that could be useful, we will use the structure equations and its differential consequences in order to find a set of algebraic generators for the EDS
\begin{align*}
  \cI_{\text{E}}&:=\left<\eta_{kp}\theta_{li}\wedge\left(\omega^{lp}+\omega^{pl}\right)+\left(-1\right)^{n+1}T^l \wedge\theta_{ikl},\Omega^{pq} \wedge\theta_{ipq}\right>_{\text{diff}}\\
  &=\left<\omega^{lp}+\omega^{pl},T^l,\Omega^{pq} \wedge\theta_{ipq}\right>_{\text{diff}}.
\end{align*}
The differential of the first set of generators $\omega^{lp}+\omega^{pl}$ yields to
\begin{align*}
  \dif\left(\omega^{lp}+\omega^{pl}\right)&=-\eta_{st}\left(\omega^{lt}\wedge\omega^{sp}+\omega^{ps}\wedge\omega^{tl}\right)+\left(\Omega^{lp}+\Omega^{pl}\right)\\
  &=-\eta_{st}\left[\left(\omega^{lt}+\omega^{tl}\right)\wedge\omega^{sp}-\omega^{tl}\wedge\left(\omega^{sp}+\omega^{ps}\right)\right]+\left(\Omega^{lp}+\Omega^{pl}\right),
\end{align*}
so the antisymmetry property for the curvature
\[
A^{lp}:=\Omega^{lp}+\Omega^{pl}=0
\]
is a differential consequence of the original Einstein equations. From the second Bianchi identity
\[
\dif T^k=\Omega^k_l \wedge\theta^l-\omega_l^k \wedge T^l
\]
another generator for the EDS $\cI_E$ is $B^k:=\Omega^k_l \wedge\theta^l$. Finally, from the last set of generators we obtain the differential consequences
\begin{align*}
  \dif\left(\Omega^{pq}\wedge\theta_{ipq}\right)&=\eta_{kl}\Omega^{pk}\wedge\left(\omega^{lq}+\omega^{ql}\right)\wedge\theta_{ipq}+\left(\omega_i^k-\delta_i^k\omega^s_s\right)\wedge\Omega^{pq}\wedge\theta_{kpq}+\Omega^{pq}\wedge T^k \wedge\theta_{ipqk};
\end{align*}
therefore there are no new algebraic generators from here. In conclusion
\[
\cI_E=\left<\omega^{lp}+\omega^{pl},\Omega^{lp}+\Omega^{pl},\dif\theta^l+\omega^l_k\wedge\theta^k,\Omega^k_l \wedge\theta^l,\Omega^{pq} \wedge\theta_{ipq}\right>_{\text{alg}}
\]
is a presentation for $\cI_E$ in terms of algebraic generators, which is an EDS for this version of gravity with vielbeins.

\subsection{A variational problem for Einstein-Palatini gravity with vielbeins}\label{SubSect:AlterVariationalPrinciple}

The expression \eqref{Eq:ConnectionVariations} for the variation of the connection taken as independent of the variation of the frame can be used to set a variational principle for GR in the sense adopted here, which loosely correspond with Einstein-Palatini gravity with vielbeins. In this case we can take as restriction EDS
\[
\cI_{PG}':=\left<\tr\omega\right>_{\text{diff}}
\]
in order to take away the projective transformations that the adopted Lagrangian form $\lambda_{PG}$ has. This restriction is analogous to the restriction considered in Eq. $(3.7)$ as appears in \cite{2012GReGr..44.2337D}; their true nature is very different, as explained above, although it serves in this case to the same purpose: To force Euler-Lagrange equations to be vacuum Einstein equations.
\\
The annihilation of the variations \eqref{Eq:ConnectionVariations} with respect to the connection yields to a set of Euler-Lagrange equations; 
namely, from
\[
\left[\eta_{kp}\theta_{li}\wedge\left(\omega^{lp}+\omega^{pl}\right)+\left(-1\right)^{n+1}T^l\wedge\theta_{ikl}\right]\wedge\delta\omega^{ik}=0
\]
we obtain the equation of motion
\begin{equation}\label{Eq:EulerLagrange1}
\eta_{kp}\theta_{li}\wedge\left(\omega^{lp}+\omega^{pl}\right)+\left(-1\right)^{n+1}T^l\wedge\theta_{ikl}=0.
\end{equation}
From the skewsymmetry of $\theta_{ijk}$ it follows that
\[
\eta_{kp}\theta_{li}\wedge\left(\omega^{lp}+\omega^{pl}\right)+\eta_{ip}\theta_{lk}\wedge\left(\omega^{lp}+\omega^{pl}\right)=0.
\]
The way to solve this system is very interesting, and uses Proposition \ref{Prop:UniqueSolution}.
\begin{lemma}
Let $x\mapsto m^{kl}$ be a set of $1$-forms on $M$ such that $\eta_{ij}m^{ij}=0$. If these forms solve the system
\[
\begin{cases}
m^{pk}\wedge\left(\eta_{pq}\theta_{kl}\pm\eta_{pl}\theta_{kq}\right)=0\cr
m^{pk}\mp m^{kp}=0
\end{cases}
\]
then
\[
m^{ij}=0.
\]
\end{lemma}
\begin{proof}
Let us define the set of local $n-1$-forms
\[
\theta_i:=X_i\lrcorner\sigma_0,
\]
where, as above, $\left\{X_1,\cdots,X_n\right\}$ is the frame dual to $\left\{\theta^1,\cdots,\theta^n\right\}$, and $\sigma_0=\theta^1\wedge\cdots\wedge\theta^n$. Then we have the identity
\[
\theta^m\wedge\theta_{kl}=-\left(\delta_k^m\theta_l-\delta^m_l\theta_k\right).
\]
Thus
\begin{align*}
m^{pk}\wedge\eta_{pq}\theta_{kl}&=\eta_{pq}m^{pk}_r\theta^r\wedge\theta_{kl}\cr
&=-\eta_{pq}m^{pk}_r\left(\delta^r_k\theta_l-\delta^r_l\theta_k\right);
\end{align*}
therefore
\begin{align*}
0&=m^{pk}\wedge\left(\eta_{pq}\theta_{kl}\pm\eta_{pl}\theta_{kq}\right)=\cr
&=-m^{pk}_r\left[\eta_{pq}\left(\delta^r_k\theta_l-\delta^r_l\theta_k\right)\pm\eta_{pl}\left(\delta^r_k\theta_q-\delta^r_q\theta_k\right)\right]\cr
&=\eta_{pq}\left(-m^{pr}_r\theta_l+m^{pk}_l\theta_k\right)\pm\eta_{pl}\left(-m^{pr}_r\theta_q+m^{pk}_q\theta_k\right).
\end{align*}
By multiplying both sides of this identity by $\theta^l$ and adding up, we see that
\begin{align*}
0&=\theta^l\wedge\left[\eta_{pq}\left(-m^{pr}_r\theta_l+m^{pk}_l\theta_k\right)\pm\eta_{pl}\left(-m^{pr}_r\theta_q+m^{pk}_q\theta_k\right)\right]\cr
&=\eta_{pq}\left(-nm^{pr}_r+m^{pr}_r\right)\sigma_0\pm\eta_{pq}\left(-m^{pr}_r\right)\sigma_0\pm \eta_{pk}m^{pk}_q\sigma_0
\end{align*}
where it was used that
\[
\theta^k\wedge\theta_l=\delta^k_l\sigma_0.
\]
Then $\eta_{pq}\left(n-1\pm 1\right)m^{pr}_r=\eta_{pk}m^{pk}_q=0$, and $m^{pr}_r=0$. Then we will have that
\begin{align*}
\left(\eta_{pq}m^{pk}_l\pm\eta_{pl}m^{pk}_q\right)\theta_k=0
\end{align*}
and from here we can conclude that the system above can be written as
\[
\begin{cases}
\eta_{pq}m^{pk}_l\pm\eta_{pl}m^{pk}_q=0\cr
m^{pk}_i\mp m^{kp}_i=0.
\end{cases}
\]
Therefore the set of unknowns $N_{ijk}:=\eta_{iq}\eta_{jp}m^{pq}_k$ solves the system
\[
\begin{cases}
N_{kql}\pm N_{klq}=0\cr
N_{kpi}\mp N_{pki}=0.
\end{cases}
\]
Using proposition \ref{Prop:UniqueSolution} we see that
\[
N_{ijk}=0
\]
is the unique solution for this system.
\end{proof}

By using the previous Lemma it follows that for a section to be an extremal section (under our choice $\text{Tr}\,\omega=0$) it will be necessary that
\[
\omega^{pl}+\omega^{lp}=0,
\]
and, as a bonus, $T=0$. Thus the generators of the EDS $\cI_{PG}$ are obtained as equations of motion.


\section{Reduction for a variational problem}\label{Sec:ReductionTheory}

An important observation concerning the variational problem
\[
\left(J^1LM\rightarrow M,\lambda_{PG},\cI_{PG}\right)
\]
is that both $\lambda_{PG}$ and $\cI_{PG}$ are $SO\left(1,n-1\right)$-invariant. It will be interesting to find a procedure in order to quotient out the degrees of freedom associated to the orbits of this symmetry group, namely, if we can apply a kind of reduction procedure, as in \cite{lopez03:_reduc}. The problem with this approach is that in this reference the authors deal with reduction of the so called \emph{classical variational problem}, namely, with variational problems of the form
\[
\left(J^1P\rightarrow M,\cL\omega,\cI_{\text{con}}\right),
\]
where $P\rightarrow M$ is a principal bundle, $\cL\in C^\infty\left(J^1P\right)$, $\omega$ is a volume form on $M$ and $\cI_{\text{con}}$ is the contact structure on $J^1P$. Therefore we must devise a reduction scheme general enough to include variational problems whose restriction EDS are different from the contact EDS of a jet space.

\subsection{Reduction of an EDS}

In order to set the reduction procedure for a variational problem, it is crucial to know how to reduce the restriction EDS. So let $M$ be a manifold, $G$ a Lie group acting on $M$ in such a way that the space of orbits $\overline{M}:=M/G$ is a manifold; we will denote by $p_G:M\rightarrow\overline{M}$ the canonical projection. Let us consider $\cI$ an EDS on $M$ such that
\[
g\cdot\cI\subset\cI\qquad\forall g\in G;
\]
the following definition can be found in \cite{MR2151123}.

\begin{definition}[Reduced EDS]
  The \emph{reduced EDS} associated to the action of $G$ on $\left(M,\cI\right)$ is the set of forms
  \[
  \overline\cI:=\left\{\alpha\in\Omega^\bullet\left(\overline{M}\right):p_G^*\alpha\in\cI\right\}.
  \]
\end{definition}

Let $\tau:P\rightarrow M$ be a $G$-principal bundle. This definition can be applied in order to reduce the contact structure on $J^1P$: It will give us an interpretation of the canonical $2$-form $\Omega_2$ on $C\left(P\right)$ as generator of the EDS on the bundle of connection obtained by reduction of the contact structure of $J^1P$, as the following example shows.

\begin{example}[Reduction of the contact structure on $J^1P$]\label{Ex:ReductionContStruct}
  Let us analyze this in more detail; the result we are looking for is local, so there is no real loss in assuming that $P=M\times G$, and this means that $J^1P=P\times_M\left(T^*M\otimes\g\right)$, by using the following correspondence: If $s:M\rightarrow P$ is a section, then
  \[
  j^1_xs=\left(x,s\left(x\right),\left(T_xs\right)\left(\cdot\right)\left(s\left(x\right)\right)^{-1}\right)
  \]
  for all $x\in M$. It means that $C\left(P\right)=T^*M\otimes\g$, and the canonical projection $q:J^1P\rightarrow C\left(P\right)$ is simply
  \[
  q\left(x,g,\xi\right)=\left(x,\xi\right).
  \]
  In these terms the $2$-form $\Omega_2$ reads
  \[
  \left.\Omega_2\right|_{\left(x,\xi\right)}=\dif\xi-\frac{1}{2}\left[\xi\stackrel{\wedge}{,}\xi\right].
  \]
  The contact structure is generated by the $1$-forms
  \[
  \left.\theta\right|_{\left(x,g,\xi\right)}:=\dif g\cdot g^{-1}-\xi;
  \]
  the $G$-action on $P$ is $h\cdot\left(x,g\right)=\left(x,gh\right)$, that lifts to $h\cdot\left(x,g,\xi\right)=\left(x,gh,\xi\right)$. Therefore a set of algebraic generators for $\cI$ is in this case
  \[
  \mathcal{G}:=\left\{\theta,\dif\theta\right\}=\left\{\dif g\cdot g^{-1}-\xi,\left(1/2\right)\left[\xi\stackrel{\wedge}{,}\xi\right]-\dif\xi\right\};
  \]
  bearing in mind future applications, the $2$-degree generator $\dif\theta$ has been written in a convenient form. Thus we have that $J^1P/G=T^*M\otimes\g$ with projection given by
  \[
  p_G\left(x,g,\xi\right)=\left(x,\xi\right).
  \]
  The quotient EDS $\bar\cI$ is graded, as $\cI$ does and $p_G^*$ is a $0$-degree morphism; if $\alpha$ is a $1$-form in $\bar\cI$, we will have that
  \[
  p^*_G\alpha=f\cdot\theta
  \]
  for some $f\in C^\infty\left(J^1P\right)$, and then if $\left(0,\zeta,0\right)\in T_{\left(x,g,\xi\right)}J^1P$, it results that
  \[
  0=\left(p_G^*\alpha\right)\left(0,\zeta,0\right)=f\left(x,g,\xi\right)\zeta.
  \]
  So $p^*_G\alpha=0$ and it means that $\alpha=0$, from the fact that $p_G$ is a surjective map. If $\beta\in\bar\cI\cap\Omega^p\left(J^1P/G\right),p>1$, we will have that
  \[
  p^*_G\beta=\mu\wedge\theta+\nu\wedge\left(\left(1/2\right)\left[\xi\stackrel{\wedge}{,}\xi\right]-\dif\xi\right)
  \]
  for some $\mu\in\Omega^{p-1}\left(J^1P/G\right)$ and $\nu\in\Omega^{p-2}\left(J^1P/G\right)$; by performing the replacement
  \[
  \dif g\cdot g^{-1}=\xi+\theta
  \]
  we can assume that neither $\mu$ nor $\nu$ have dependence in the $g$-variable. Therefore, by contracting this identity with an infinitesimal generator for the $G$-action, we obtain that
  \begin{align*}
    0&=\left(0,\zeta,0\right)\lrcorner\left(p_G^*\beta\right)\cr
    &=\left(-1\right)^{p+1}\zeta\mu
  \end{align*}
  and so $\mu=0$; from here we can conclude that
  \[
  \bar\cI=\left<\frac{1}{2}\left[\xi\stackrel{\wedge}{,}\xi\right]-\dif\xi\right>_{\text{alg}},
  \]
  meaning that the reduced EDS is generated by $\Omega_2$.
\end{example}

The previous example gives some insight in the subtleties concerning the reduction of an EDS: The original contact structure is locally generated by $1$-forms, but the reduced EDS is generated by a collection of $2$-forms. Nevertheless, there exists a result allowing us to find a set of generators for a reduced EDS, under mild conditions, namely, by requiring the generators to be pullback of some forms along a projection. It is convenient to note that it was not fulfilled in the previous example, because $\dif g\cdot g^{-1}-\xi$ is not the pullback along $p_G$ of any form on $M\times\g$.  

\begin{proposition}\label{Prop:ReducingFiberedEDS}
  Let $p:M\rightarrow N$ be a fibration and $\cI:=\left<\alpha_1,\cdots,\alpha_p\right>_{\text{diff}}$ a differential ideal such that $\alpha_i\in\Omega^{k_i}\left(M\right)$ for some integers $k_i$. Let us suppose that on $N$ there exists a set of forms $\left\{\beta_1,\cdots,\beta_p\right\}$ such that
  \[
  \alpha_i=p^*\beta_i\qquad\text{for }i=1,\cdots,p.
  \]
  Then $\overline\cI=\left<\beta_1,\cdots,\beta_p\right>_{\text{diff}}$.
\end{proposition}
\begin{proof}
  The inclusion $\left<\beta_1,\cdots,\beta_p\right>_{\text{diff}}\subset\overline\cI$ follows from the definition of reduced EDS. On the other side, if $\omega\in\overline\cI$, there exists $\gamma_1,\cdots,\gamma_p$ such that
  \[
  p^*\omega=\gamma_1\wedge p^*\beta_1+\cdots+\gamma_p\wedge p^*\beta_p.
  \]
  Thus it is enough to prove that it implies $\gamma_i=p^*\sigma_i$ for all $i=1,\cdots,p$.
\end{proof}

We are now ready to introduce a generalization of the usual scheme of reduction suitable for our version of (generalized) variational problems; as far as we know, it is an original contribution made in the present work.

\begin{definition}[Reduction of a variational problem]
  Let $\left(\Lambda,\lambda,\cI\right)$ be a variational problem on the bundle $p:\Lambda\rightarrow M$. Let us suppose that a Lie group $G$ acts on $\Lambda$ such that
  \begin{enumerate}
  \item the action is free and proper,
  \item its orbits are vertical, i.e. $p\left(g\cdot u\right)=p\left(u\right)$ for all $u\in\Lambda$ and $g\in G$,
  \item there exists $\bar\lambda\in\Omega^n\left(\bar{M}\right)$ such that $p_G^*\bar\lambda=\lambda$, where $p_G:\Lambda\rightarrow\overline\Lambda:=\Lambda/G$ is the canonical projection, and
  \item it is a symmetry group for the EDS $\cI$.
  \end{enumerate}
  The \emph{reduced variational problem} for $\left(\Lambda,\lambda,\cI\right)$ is the variational problem $\left(\overline\Lambda,\overline\lambda,\overline\cI\right)$, where $\overline\cI$ is the reduced EDS for $\cI$.
\end{definition}

\begin{example}[Euler-Poincaré reduction]
  The Euler-Poincaré reduction \cite{lopez00:_reduc_princ_fiber_bundl,castrillon12:_const_euler_poinc_reduc_field_theor} can be seen as an instance of this reduction scheme. In the setting of Example \ref{Ex:ReductionContStruct}, we see that the reduced variational problem associated to $\left(J^1P,L\omega,\cI_{\text{con}}\right)$, where $L\in C^\infty\left(J^1P\right)^G$ and $\omega$ is an invariant volume, is nothing but $\left(C\left(P\right),\overline{L}\omega,\left<\Omega_2\right>_{\text{diff}}\right)$. The relationship with Euler-Poincaré reduction can be revealed by means of the following consideration: The restrictions on the possible variations of the fields of the reduced field theory are exactly those defining the allowed variations (see Definition \ref{Def:AdmVars}), namely, the infinitesimal variations for the restriction EDS. If $\sigma:U\subset M\rightarrow C\left(P\right):x\mapsto\left(x,\xi\left(x\right)\right)$ is a local section of the bundle of connections, integral for $\left<\Omega_2\right>_{\text{diff}}$, then a vertical vector field $\left(0,\Xi\right)$ will be an infinitesimal symmetry of $\sigma$ iff
  \[
  \sigma^*\left(\dif\Xi-\left[\Xi,\xi\right]\right)=0.
  \]
  Let $\ad{P}:=\left(P\times\g\right)/G$ be the adjoint bundle associated to $P$. By using the fact that $C\left(P\right)$ is an affine bundle modelled on $T^*M\otimes\ad{P}$, we can identify any variation $\Xi$ of its sections as an $\ad{P}$-valued $1$-form on $M$. In these terms the requeriment of admissibility for variations reads $\dif_\sigma\Xi=0$, where $\dif_\sigma$ is the covariant exterior differential associated to the connection $\sigma$; by using that $\sigma$ is flat, then (at least locally) there exists a section $\eta:U\subset M\rightarrow\ad{P}$ for the adjoint bundle such that
  \[
  \Xi=\dif_\sigma\eta.
  \]
  For $H$ an arbitrary connection, we see from here that
  \[
  \Xi=\dif_{H+\sigma-H}\eta=\dif_H\eta+\left[\sigma-H,\eta\right],
  \]
  the usual requeriment for variations in Euler-Poincaré reduction (compare with Prop.\ 3.1 in \cite{lopez00:_reduc_princ_fiber_bundl}).
\end{example}

\subsection{Reduced gravity with vielbeins}

We are ready to perform the reduction of the variational problem $\left(J^1LM,\lambda_{PG},\cI_{PG}\right)$ by the Lorentz subgroup $H:=SO\left(1,n-1\right)$. The $H$-action is readily seen to be free and proper, and its orbits are vertical, so it remains to verify that $\lambda_{PG}$ is $H$-horizontal and $\cI_{PG}$ is $H$-invariant.

In order to properly show this invariance, we need to introduce a nice description for the quotient bundle
\[
\begin{array}{rcl}
  \tau:J^1LM/H&\rightarrow&M\\
  \left[j^1_xs\right]_H&\mapsto&x.
\end{array}
\]
Now let us remember that the bundle $J^1LM\rightarrow J^1LM/GL\left(n\right)$ is isomorphic to the pullback bundle $p^*J^1LM$, where $p:J^1LM/GL\left(n\right)\rightarrow M$ is the projection induced by $\tau_1:J^1LM\rightarrow M$; from this perspective every element $j^1_xs$ of $J^1LM$ can be written as a pair
\[
j^1_xs=\left(\left[j^1_xs\right]_{GL\left(n\right)},s\left(x\right)\right),
\]
and the $GL\left(n\right)$-action is simply
\[
\left(\left[j^1_xs\right]_{GL\left(n\right)},s\left(x\right)\right)\cdot g=\left(\left[j^1_xs\right]_{GL\left(n\right)},s\left(x\right)\cdot g\right).
\]
Then we have the following representation for the quotient $J^1LM/H$.
\begin{lemma}
  Let $\left[\tau\right]:\Sigma\rightarrow M$ be the bundle of metrics on $M$. Then $J^1LM/H$ is isomorphic to the pullback bundle $p^*\Sigma=C\left(LM\right)\times_M\Sigma$.
\end{lemma}

The local version for these results is very illuminating of the geometrical meaning of the sections of these bundles.

\begin{proposition}
  Let $\left(x^\mu,e_k^\nu,e^\sigma_{j\mu}\right)$ be the set of jet coordinates introduced above on $J^1LM$. Then there exists a set of coordinates $\left(x^\mu,g^{\mu\nu},\Gamma^\sigma_{\rho\gamma}\right)$ on $p^*\Sigma$ such that
  \[
  \tilde{g}^{\mu\nu}:=g^{\mu\nu}\circ p_H=\eta^{kl}e_k^\mu e_l^\nu,\qquad\tilde{\Gamma}^\sigma_{\rho\gamma}:=\Gamma^\sigma_{\rho\gamma}\circ p_H=-e^\sigma_{k\gamma}e^k_\rho.
  \]
  In terms of these functions
  \begin{equation}
    {\lambda_{PG}}=\epsilon_{\mu_1\cdots\mu_{n-2}\gamma\kappa}\sqrt{-\det{\tilde{g}}}\tilde{g}^{\kappa\phi}\dif x^{\mu_1}\wedge\cdots\wedge\dif x^{\mu_{n-2}}\wedge\left(\dif\tilde\Gamma^\gamma_{\rho\phi}\wedge\dif x^\rho+\tilde\Gamma^\sigma_{\delta\phi}\tilde\Gamma^\gamma_{\beta\sigma}\dif x^\beta\wedge\dif x^\delta\right)
  \end{equation}
  and
  \begin{align}
    \eta^{ik}\omega_k^j+\eta^{jk}\omega_k^i&=e^i_\mu e^j_\nu\left(\dif \tilde{g}^{\mu\nu}+\left(\tilde{g}^{\mu\sigma}\tilde{\Gamma}^\nu_{\gamma\sigma}+\tilde{g}^{\nu\sigma}\tilde{\Gamma}^\mu_{\gamma\sigma}\right)\dif x^\gamma\right)\\
    T^i&=e^i_\sigma\tilde{\Gamma}^\sigma_{\mu\nu}\dif x^\mu\wedge\dif x^\nu\\
    \Tr\omega&=g_{\mu\nu}\dif g^{\mu\nu}+\Gamma^\sigma_{\sigma\rho}\dif x^\rho.\label{Eq:TraceEqReduced}
  \end{align}
  In particular, the Lagrangian form $\lambda_{PG}$ is horizontal for the $H$-projection, and the EDS $\cI_{PG}$ is $H$-invariant.
\end{proposition}

By combining these equations and Proposition \ref{Prop:ReducingFiberedEDS} we deduce the following corollary.

\begin{corollary}
  The reduced Lagrangian $\overline{\lambda}_{PG}$ is given by
  \[
    \overline{\lambda_{PG}}=\epsilon_{\mu_1\cdots\mu_{n-2}\gamma\kappa}\sqrt{-\det{g}}g^{\kappa\phi}\dif x^{\mu_1}\wedge\cdots\wedge\dif x^{\mu_{n-2}}\wedge\left(\dif\Gamma^\gamma_{\rho\phi}\wedge\dif x^\rho+\Gamma^\sigma_{\delta\phi}\Gamma^\gamma_{\beta\sigma}\dif x^\beta\wedge\dif x^\delta\right),
  \]
  and the reduced EDS $\overline\cI_{PG}$ can be generated as
  \[
  \overline\cI_{PG}=\left<\dif g^{\mu\nu}+\left(g^{\mu\sigma}\Gamma^\nu_{\gamma\sigma}+g^{\nu\sigma}\Gamma^\mu_{\gamma\sigma}\right)\dif x^\gamma,\Gamma^\sigma_{\mu\nu}-\Gamma^\sigma_{\nu\mu}\right>_{\text{diff}}.
  \]
\end{corollary}

\begin{note}[Levi-Civita EDS]
  By recalling that $J^1LM/H$ is the product bundle $C\left(LM\right)\times_M\Sigma$, the reduced EDS $\overline\cI_{PG}$ can be interpreted geometrically: For every metric $g:M\rightarrow\Sigma$ on $M$, the unique connection $\Gamma:M\rightarrow C\left(LM\right)$ such that $\Gamma\times g:M\rightarrow C\left(LM\right)\times_M\Sigma$ is an integral section for $\overline\cI_{PG}$ is the Levi-Civita connection for $g$. Thus we can call this EDS the \emph{Levi-Civita EDS}.
\end{note}

\subsection{Discussion: Einstein gravity as a reduced variational problem for Palatini gravity}

The reduced variational problem of the variational problem describing Einstein-Hilbert gravity with vielbeins can be considered as equivalent to the Einstein-Hilbert (i.e., without vielbeins!) variational problem. In fact, we can consider the following diagram
\[
\begin{diagram}
  \node[2]{C\left(LM\right)\times_M J^1\Sigma}\arrow[2]{s,r}{}\arrow{se,t}{\Pi}\arrow{sw,t}{p_2}\\\node{J^1\Sigma}\arrow{se,b}{}\node[2]{C\left(LM\right)\times_M\Sigma}\arrow{sw,b}{}\\
  \node[2]{M}
\end{diagram}
\]
induced by $J^1\Sigma\rightarrow\Sigma$; let us define
\[
\cJ:=\left<p_2^*\cI_{\text{con}}^\Sigma,\Pi^*\overline\cI_{PG}\right>_{\text{diff}},
\]
where $\cI^\Sigma_{\text{con}}$ is the contact structure on $J^1\Sigma$. Then the next result follows.
\begin{lemma}
  There exists a manifold $L\subset C\left(LM\right)\times_M J^1\Sigma$ minimal with respect to the property that every integral manifold of $\cJ$ must be included in it. The map $\left.p_2\right|_L:L\rightarrow J^1\Sigma$ is a bundle isomorphism such that $p_2^*\cI_{\text{con}}^\Sigma=\left.\cJ\right|_L$.
\end{lemma}

Then the $\Pi-$projectable extremals of
\[
\left(L,\left.\Pi^*\overline\lambda_{PG}\right|_L,\left.\cJ\right|_L\right)
\]
are in one-to-one correspondence with the extremals of the reduced Palatini variational problem on $C\left(LM\right)\times_M\Sigma$ via $\Pi$, and with the extremals of the classical variational problem $\left(J^1\Sigma,\lambda_{EH},\cI^\Sigma_{\text{con}}\right)$ through $\left.p_2\right|_L$; the Einstein-Hilbert Lagrangian $\lambda_{EH}$ is determined by the equation
\[
p_2^*\lambda_{EH}=\left.\Pi^*\overline\lambda_{PG}\right|_L.
\]
It induces the equivalence we were looking for.
\\
It remains to consider reduction of the variational problem $\left(J^1LM,\lambda_{PG},\cI_{PG}'\right)$ considered in Subsection \ref{SubSect:AlterVariationalPrinciple}, corresponding to the Einstein-Palatini gravity with vielbeins; because the form $\Tr\omega$ is invariant by the action of the Lorentz group (in fact, it is invariant by the full general linear group $GL\left(n\right)$), the reduction must be done by using the same group. By means of Eq. \eqref{Eq:TraceEqReduced} it can be concluded that the reduced variational problem is $\left(C\left(LM\right)\times_M\Sigma,\overline{\lambda_{PG}},g_{\mu\nu}\dif g^{\mu\nu}+\Gamma^\sigma_{\sigma\rho}\dif x^\rho\right)$; it is a kind of Einstein-Palatini formulation for gravity, although different from \cite{citeulike:820116}.

\section{Conclusions}

In this work a geometrical formulation for Palatini gravity was provided, by using a broader notion for the term \emph{variational problem}. In order to perform this task, it was necessary to use some constructions associated to the jet space of the frame bundle. This picture would give us some insights on the geometrical character of vacuum GR, complementary to those found in the literature. In order to relate this formulation with the usual Einstein-Hilbert variational problem, a generalized reduction scheme was set.

\appendix

\section{Notations}

The internal metric of the tetrads will have the signature $\left(-+\cdots+\right)$. We will assume further the conventions of \cite{MarsdenRatiu} in playing with forms. If $\alpha$ is a $k-$form:
\begin{align*}
&X\lrcorner\left(\alpha\wedge\beta\right)=\left(X\lrcorner\alpha\right)\wedge\beta+\left(-1\right)^k\alpha\wedge\left(X\lrcorner\beta\right)\cr
&d\left(\alpha\wedge\beta\right)=d\alpha\wedge\beta+\left(-1\right)^k\alpha\wedge
d\beta\cr
&\left(d\alpha\right)\left(X_0,\cdots,X_k\right)=\sum_{i=0}^k\left(-1\right)^iX_i\cdot\left(\alpha\left(X_0,\cdots,\widehat{X_i},\cdots,X_k\right)\right)+\cr
&\qquad+\sum_{0\leq
i<j\leq k}\left(-1\right)^{i+j}\alpha\left(\left[X_i,X_j\right],X_0,\cdots,\widehat{X_i},\cdots,\widehat{X_j},\cdots,X_k\right)
\end{align*}
The indices $\mu,\nu,\cdots$ and $i,j,k,\cdots$ will run from $1$ to $n$; as usual, the first set will be used in the enumeration of local coordinates on spacetime, while the latin indices will label the components in the (tensorial algebra of the) local model $\mR^n$. In particular, we are using the following convention relating group product in $GL\left(n\right)$ and indices
\[
\left(g\cdot h\right)_i^j=g_i^kh_k^j
\]
for all $g,h\in GL\left(n\right)$. Following standard usage, we will use the acronym EDS when refering to \emph{exterior differential systems}.

\section{Some geometrical results}

\subsection{An important algebraic result} We would like to state here the following algebraic proposition.

\begin{proposition}\label{Prop:UniqueSolution}
Let $\left\{c_{ijk}\right\}$ be a set of real numbers such that
\[
\begin{cases}
c_{ijk}\mp c_{jik}=b_{ijk}\cr
c_{ijk}\pm c_{ikj}=a_{ijk}
\end{cases}
\]
for some given set of real numbers $\left\{a_{ijk}\right\}$ and $\left\{b_{ijk}\right\}$ such that $b_{ijk}\mp b_{jik}=0$ and $a_{ijk}\pm a_{ikj}=0$. Then
\[
c_{ijk}=\frac{1}{2}\left(a_{ijk}+a_{jki}-a_{kij}+b_{ijk}+b_{kij}-b_{jki}\right)
\]
is the unique solution for this linear system.
\end{proposition}
\begin{proof}
From first equation we see that
\[
\pm c_{jik}=c_{ijk}-b_{ijk}.
\]
The trick now is to form the following combination
\begin{align*}
a_{ijk}+a_{jki}-a_{kij}&=c_{ijk}\pm c_{ikj}+c_{jki}\pm c_{jik}-\left(c_{kij}\pm c_{kji}\right)\cr
&=2c_{ijk}-b_{ijk}-b_{kij}+b_{jki}
\end{align*}
where in the permutation of indices was used the remaining condition.
\end{proof}

\subsection{Variational problems and field theory}\label{Sec:VarProb+Dirac}
Our initial data will be a \emph{variational triple}, that is, a triple $\left(\Lambda\rightarrow M,\lambda,\cI\right)$ composed of a fibre bundle $\Lambda$ on the spacetime $M$, a $n$-form $\lambda$ on it (where $n=\text{dim}M$) and an EDS $\cI\subset\Omega^\bullet\left(\Lambda\right)$. The bundle consists of the degree of freedom associated to the fields and its (generalized) velocities, the $n$-form $\lambda$ will be used to define the dynamics, and $\cI$ will induce some relations between the degrees of freedom (for example, it will force to some variables to be the derivatives of another variables).

\subsubsection{Variational problems}
We are in position to formulate a notion of \emph{variational problems}. Although we are primarily interested in applications of these notions to physics, they can be used in tackling geometrical problems, see \cite{hsu92:_calcul_variat_griff}. 
\begin{definition}\label{Def:VariationalProblems}
  The \emph{variational problem} associated to a variational triple $\left(\Lambda,\lambda,\cI\right)$ consists into the problem of finding the sections $\sigma:M\rightarrow\Lambda$ which are integrals for the EDS $\cI$ and extremals for the functional
\[
S\left[\sigma\right]:=\int_M\sigma^*\lambda.
\]
\end{definition}

\begin{note}
We will suppose that the necessary conditions for the existence of the several integrals that could appear throughout the work are met; for example, $M$ would be compact.
\end{note}

\begin{definition}[Infinitesimal symmetries of an EDS]
  Let $\cI\subset\Omega^\bullet\left(\Lambda\right)$ be an EDS. A (perhaps local) vector field $X$ is an \emph{infinitesimal symmetry of $\cI$} if and only if
  \[
  \cL_X\cI\subset\cI.
  \]
  The set of infinitesimal symmetries of $\cI$ will be indicated by $\Sym{\cI}$.
\end{definition}

\begin{definition}[Euler-Lagrange EDS]
  Let $\left(\Lambda,\lambda,\cI\right)$ be a variational problem. The \emph{Euler-Lagrange EDS} is the EDS generated by the set of forms
  \[
  \left\{\alpha\in\Omega^\bullet\left(\Lambda\right):\alpha-X\lrcorner\dif\lambda\equiv0\mod{\dif\Omega^{n-1}\left(\Lambda\right)}\text{ for all }X\in\Sym{\cI}\cap\mathfrak{X}^V\left(\Lambda\right)\right\}.
  \]
\end{definition}

\subsubsection{Classical field theory as a variational problem}\label{SubSubSec:FieldTheoryAsVarProblem}
  It is necessary perhaps to indicates the way in which the usual (first order) classical field theory fits in this scheme: The corresponding variational problem is simply $\left(J^1E,\cL\dif x^1\wedge\cdots\wedge\dif x^n,\cI_{\text{con}}\right)$, where $E\rightarrow M$ is a bundle on $M$ (whose nature is associated to the field to be described by the theory), $\cL$ is the Lagrangian density of the theory and
  \[
  \cI_{\text{con}}:=\left<\dif u^A-u^A_k\dif x^k\right>_{\text{diff}}
  \]
  is the contact structure of the jet space. This variational problem is usually called in the literature the \emph{classical variational problem} \cite{GotayCartan,book:852048}. Then we have the following result.
  \begin{lemma}
    The underlying PDE for the Euler-Lagrange EDS associated to the classical variational problem contains the Euler-Lagrange equations.
  \end{lemma}
  \begin{proof}
    Let us work in local coordinates. For $X:=\left(0,\delta u^A,\delta u^A_k\right)\in\Sym{\cI_{\text{con}}}\cap \mathfrak{X}^V\left(\Lambda\right)$, we have that
    \[
    \dif\delta u^A-\delta u^a_k\dif x^k=0;
    \]
    then
    \[
    X\lrcorner\dif\lambda=\left(\frac{\partial\cL}{\partial u^A}\delta u^A+\frac{\partial\cL}{\partial u^A_k}\delta u^A_k\right)\dif x^1\wedge\cdots\wedge\dif x^n.
    \]
    Let $\sigma_A\in\Omega^{n-1}\left(J^1E\right)$ be defined as
    \[
    \sigma_A:=\left(\frac{\partial\cL}{\partial u^A_k}\frac{\partial}{\partial x^k}\right)\lrcorner\dif x^1\wedge\cdots\wedge\dif x^n;
    \]
    therefore
    \begin{align*}
      X\lrcorner\dif\lambda&=\frac{\partial\cL}{\partial u^A}\delta u^A\dif x^1\wedge\cdots\wedge\dif x^n+\sigma_A\wedge\delta u^A_k\dif x^k\\
      &=\frac{\partial\cL}{\partial u^A}\delta u^A\dif x^1\wedge\cdots\wedge\dif x^n+\sigma_A\wedge\dif\delta u^A\\
      &\equiv\left[\left(-1\right)^{n+1}\dif\sigma_A+\frac{\partial\cL}{\partial u^A}\dif x^1\wedge\cdots\wedge\dif x^n\right]\delta u^A\mod\dif\Omega^{n-1}\left(J^1E\right),
    \end{align*}
    and so
    \[
    \alpha_A:=\left(-1\right)^{n+1}\dif\sigma_A+\frac{\partial\cL}{\partial u^A}\dif x^1\wedge\cdots\wedge\dif x^n
    \]
    are generators for the Euler-Lagrange EDS. Any integral section for this EDS will obey the Euler-Lagrange equations associated to $\cL$.
  \end{proof}

\subsection{Some tools from differential geometry: The geometry of $J^1P$}\label{SubSection:GeometryJP}
It is time to introduce the basic language we will use to describe gravitation in this work; it will be necessary to point out here some useful tools borrowed from differential geometry in the handling of the multiple questions raised while working with connections. 

In the present section we will describe briefly differential geometry from moving frame viewpoint, as in \cite{KN1,MR532831}. Whenever possible, we will make contact with the more usual description in terms of principal bundles; this framework is of outmost importance in the description of Palatini gravity in the present work. So we will need some facts concerning the jet bundle of a principal bundle. This is a natural choice in this context, because the first structure equation on $LM$ allow us to consider a connection as a kind of velocity associated to the degree of freedom provided by a frame. From this point of view, we need a set of forms on $J^1LM$ encoding the structure equations; it results from the work of García \cite{MR0315624} and Castrillón \emph{et al.} \cite{springerlink:10.1007/PL00004852} that there exists a $\gl\left(n\right)$-valued $2$-form on $J^1LM$ such that its pullback along a connection (in an appropiate sense, see below for details) is the curvature of this connection. Additionally, it can be defined a $\mR^n$-valued $2$-form giving rise to the torsion of the connection via the same pullback procedure. These forms are the fundamental ingredients in the construction of the equivalent of Palatini Lagrangian in this context.

The following section has been formulated by making heavy use of the reference \cite{springerlink:10.1007/PL00004852}; it can be considered as a natural continuation of the last example in \cite{martinez04:_class_field_theor_lie_variational} to this context.

\subsubsection{Geometric preliminaries}
Let $\tau:P\rightarrow M$ be a $G$-principal bundle on $M$; then we have the diagram
\begin{equation}\label{Eq:DiagramAtiyah}
\begin{diagram}
  \node{TP/G}\arrow{s,l}{\tau^P_M}\arrow{e,t}{T\tau}\node{TM}\arrow{s,r}{\tau_M}\\
  \node{M}\arrow{e,b,=}{}\node{M}
\end{diagram}
\end{equation}
where it was defined
\[
\tau^P_M\left(\left[v\right]_G\right):=\tau\left(\tau_P\left(v\right)\right).
\]
\begin{definition}
  The bundle of connections $C\left(P\right)$ is the bundle on $M$ given by
\[
C\left(P\right):=\left\{\lambda:T_mM\rightarrow\left.\left(TP/G\right)\right|_m\text{ such that }T\tau\circ\lambda=\id_{T_mM.}\right\}.
\]
\end{definition}

The main tool to work with this bundle is the following lemma; it relies essentially in the fact that the $G$-orbits are vertical, and the action is free.

\begin{lemma}
  There exists a bundle isomorphism between $C\left(P\right)$ and $J^1P/G$.
\end{lemma}

The bundle isomorphism between $C\left(P\right)$ and $J^1P/G$ is proved in \cite{springerlink:10.1007/PL00004852} using the fact that there exists a $G$-principal bundle structure $q:J^1P\rightarrow C\left(P\right)$ through the right $G$-action determined by the lift of the $G$-action on $P$. If we consider the $1$-jet space as the set
\[
J^1P=\bigcup_{p\in P}\left\{\rho:T_{\tau\left(p\right)}M\rightarrow T_pP\text{ such that }T_p\tau\circ\rho=\id_{T_{\tau\left(p\right)}M}\right\},
\]
then $q\left(\rho\right):=p_G\circ\rho$, where $p_G:TP\rightarrow TP/G$ is the canonical projection for the quotient; it is convenient at this point to remember that the $1$-jet bundle $J^1P$ comes with the maps fitting in the diagram
\[
\begin{diagram}
  \node{J^1P}\arrow[2]{e,t}{\tau_{10}}\arrow{se,b}{\tau_1}\node[2]{P}\arrow{sw,b}{\tau}\\
\node[2]{M}
\end{diagram}
\]
This identification allows us to use the map $\tau_1:J^1P\rightarrow M$ as the fibre bundle map of $C\left(P\right)$ on $M$. On the other side, every element $\left[\rho\right]_G\in C\left(P\right)$ can be thought as a ``connection form at $m:=\tau_1\left(\left[\rho\right]_G\right)$'', as the following proposition shows.
\begin{proposition}\label{Prop:CElementsAsLocalConnections}
  Every element $\left[\rho\right]_G$ defines a unique family of projections $\Gamma_p:T_pP\rightarrow V_pP$ for $p\in\tau^{-1}\left(m\right)$.
\end{proposition}
\begin{proof}
In fact, for $p\in\tau^{-1}\left(m\right)$, we can define the projection map
\[
\Gamma_p:=T_{\tau_{10}\left(\rho\right)}R_g\circ\Gamma_{\tau_{10}\left(\rho\right)}\circ T_{p}R_{g^{-1}},
\]
if and only if $p=\tau_{10}\left(\rho\right)g$ and
\[
\Gamma_{\tau_{10}\left(\rho\right)}:=\id_{T_{\tau_{10}\left(\rho\right)}P}-\rho\circ T_{\tau_{10}\left(\rho\right)}\tau.
\]
Namely, we select an element $\rho\in\left[\rho\right]_G$ and define on $p_0:=\tau_{10}\left(\rho\right)\in P$ a projection $\Gamma_{p_0}:T_{p_0}P\rightarrow V_{p_0}P$ onto the vertical fibre; then we extend this projection to any point of $\tau^{-1}\left(m\right)$ by using the right $G$-action. Because of the form we choose to do this extension, it results that this definition is independent of the choice made of the representative $\rho\in\left[\rho\right]_G$.
\end{proof}

\subsubsection{The universal form on $J^1P$}

We can use Proposition \ref{Prop:CElementsAsLocalConnections} in order to construct a $\g$-valued $1$-form on $J^1P$; it is necessary first to recall that we have a vector bundle isomorphism
\[
P\times\g\rightarrow VP:\left(p,\xi\right)\mapsto\left.\frac{\vec{\text{d}}}{\text{d}t}\right|_{t=0}\left[p\cdot\left(\exp{t\xi}\right)\right].
\]
Then the canonical connection $\omega\in\Omega^1\left(J^1P,\g\right)$ is defined through
\[
\left.\omega\right|_\rho\left(Y\right):=\left[\rho\right]_G\left(T_\rho\tau_{10}\left(Y\right)\right)
\]
where $\left[\rho\right]_G\in C\left(P\right)$ has been considered in this formula as the family of projections $\Gamma_{\tau_{10}\left(\rho\right)}$ of the previous proposition, and we freely use the identification $V_{\tau_{10}\left(\rho\right)}P\simeq\g$.
\begin{lemma}
  The $1$-form $\omega$ generates the contact structure on $P$.
\end{lemma}
\begin{proof}
  Let $s:M\rightarrow P$ be a local section of $P$; then we have that
\[
\tau_{10}\circ j^1s=s
\]
where
\[
j^1s:M\rightarrow J^1P:m\mapsto\left[v\in T_mM\mapsto\left(T_ms\right)\left(v\right)\right]
\]
is the $1$-jet of $s$. So if $v\in T_mM$ and $Y=T_m\left(j^1s\right)\left(v\right)$, we will have that
\[
T_{j^1_ms}\tau_{10}\left(Y\right)=T_ms\left(v\right);
\]
therefore
\begin{align*}
  \left.\omega\right|_{j^1_ms}\left(Y\right)&=\left[j^1_ms\right]_G\left(T_ms\left(v\right)\right)\\
&=T_ms\left(v\right)-j^1_ms\circ T_{s\left(m\right)}\tau\left(T_ms\left(v\right)\right)\\
&=T_ms\left(v\right)-j^1_ms\left(v\right)\\
&=0.
\end{align*}
It means that $\omega$ is in the algebraic closure of the contact forms.
\end{proof}
Thus we have the following result \cite{springerlink:10.1007/PL00004852}.
\begin{proposition}
  The bundle $p_G:J^1P\rightarrow J^1P/G=C\left(P\right)$ is a $G$-principal bundle, and $\omega$ defines a connection on it.
\end{proposition}
We can form now the pullback bundle
\[
\begin{diagram}
  \node{\tau_1^*\left(\ad P\right)}\arrow{s,l}{p_1}\arrow{e,t}{p_2}\node{\ad P}\arrow{s,r}{}\\
  \node{C\left(P\right)}\arrow{e,b}{\tau_1}\node{M}
\end{diagram}
\]
where $\ad P:=\left(P\times\g\right)/G$, taking $\g$ as a $G$-space through the adjoint action.
Then there exists a $\g$-valued $2$-form $\Omega$ of the adjoint type on $J^1P$, namely the curvature form associated to the connection $\omega$; it defines a $2$-form on $C\left(P\right)$ with values in $\tau_1^*\left(\ad P\right)$ via
\[
\left.\Omega_2\right|_{\left[\rho\right]_G}\left(T_{\rho}\tau_1\left(X\right),T_{\rho}\tau_1\left(Y\right)\right):=\left[\rho,\left.\Omega\right|_\rho\left(X,Y\right)\right]_G
\]
for $X,Y\in T_{\rho}J^1P$.

\subsubsection{The form $\omega$ as a universal connection}

We will prove here that the form $\omega$ can be considered as a ``universal form'', namely, that every connection on $P$ can be built as a pullback of it along a suitable map.

\begin{proposition}\label{Prop:IsoBetweenPpalBundles}
  The $G$-principal bundle $p_G:J^1P\rightarrow J^1P/G$ is isomorphic to $p^*P$.
\end{proposition}
\begin{proof}
  The bundle $p^*P$ is defined through the diagram
  \[
  \begin{diagram}
    \node{p^*P}\arrow{s,l}{p_1}\arrow{e,t}{p_2}\node{P}\arrow{s,r}{\tau}\\
    \node{J^1P/G}\arrow{e,b}{p}\node{M}
  \end{diagram}
  \]
  where $p_1,p_2$ are the projections onto the factors of the cartesian product ($p^*P\subset J^1P/G\times P$), and
  \[
  p\left(\left[j^1_xs\right]_G\right):=x.
  \]
  The bundle isomorphism is defined by
  \[
  j^1_xs\in J^1P\mapsto\left(\left[j^1_xs\right]_G,s\left(x\right)\right)\in J^1P/G\times P,
  \]
  whose range is in $p^*P$, because
  \[
  p\left(\left[j_x^1s\right]_G\right)=x=\tau\left(s\left(x\right)\right).
  \]
  In order to show that it is a diffeomorphism, it is enough to show an inverse map, namely
  \[
  \left(\left[j_x^1s\right]_G,u\right)\in p^*P\mapsto j^1_x\tilde{s}
  \]
  where $\tilde{s}:U_x\subset M\rightarrow P$ is a local section for $P$ defined in a neighborhood $U_x$ of $x$ such that $\tilde{s}\left(x\right)=u$ and
  \[
  \left[j^1_x\tilde{s}\right]_G=\left[j^1_xs\right]_G.
  \]
  It is clear that such a section exists, by defining $\tilde{s}\left(y\right)=s\left(y\right)\cdot g_0$ for $y\in U_x$ and $g_0\in G$ such that $u=s\left(x\right)\cdot g_0$. Moreover, it is a well-defined map, and this can be seen as follows: If $\tilde{t}$ is another local section verifying that $\tilde{t}\left(x\right)=u$, then there exists $\gamma:U_x\rightarrow G$ such that
  \[
  \tilde{t}\left(y\right)=s\left(y\right)\cdot\gamma\left(y\right)
  \]
  for all $y\in U_x$, and in particular $\gamma\left(x\right)=g_0$; so
  \begin{equation}\label{Eq:JSTildeJS}
    j_x^1\tilde{t}:v\in T_xM\rightarrow T_x\tilde{s}\left(v\right)=T_xs\left(v\right)\cdot g_0+\left[\left(L_{g_0*}\right)^{-1}T_x\gamma\left(v\right)\right]^P_u\in T_uP,
  \end{equation}
where $\xi^P_u\in V_uP$ indicates the infinitesimal generator for the $G$-action on $P$ associated to the element $\xi\in\g$. Therefore from the condition $\left[j^1_x\tilde{t}\right]_G=\left[j^1_xs\right]_G$ we obtain that
  \[
  T_x\tilde{t}\left(v\right)=T_xs\left(v\right)\cdot g_1
  \]
  for some $g_1\in G$, and thus must be $g_1=g_0$, because $T_xs\left(v\right)\cdot g_1\in T_{s\left(x\right)\cdot g_1}$ and $T_x\tilde{t}\left(v\right)\in T_uP=T_{s\left(x\right)\cdot g_0}P$. Therefore
  \[
  \left[\left(L_{g_0*}\right)^{-1}T_x\gamma\left(v\right)\right]^P_u=0
  \]
  and using Eq. \eqref{Eq:JSTildeJS}, $j_x^1\tilde{t}=j_x^1s\cdot g_0=j_x^1\tilde{s}$.
\end{proof}

According to Prop. \ref{Prop:CElementsAsLocalConnections}, every section of $p:C\left(P\right)\rightarrow M$ defines a connection on $P$ and conversely, every connection gives rise to a section of the bundle $C\left(P\right)$; by using Prop. \ref{Prop:IsoBetweenPpalBundles}, we can state the following result.
\begin{proposition}
Every connection $\Gamma$ determines a section of the affine bundle $\tau_{10}:J^1P\rightarrow P$.
\end{proposition}
\begin{proof}
  We will denote by $\sigma_\Gamma:M\rightarrow C\left(P\right)$ the section associated by Prop. \ref{Prop:CElementsAsLocalConnections} to the connection $\Gamma$; thus we define the map
  \[
  \tilde{\sigma}_\Gamma\left(u\right):P\rightarrow C\left(P\right)\times P:u\mapsto\left(\sigma_\Gamma\left(\tau\left(u\right)\right),u\right).
  \]
  But it is immediate to show that its range is in $p^*P$, because of the identity
  \[
  p\left(\sigma_\Gamma\left(u\right)\right)=\tau\left(u\right);
  \]
  it is additionally a section of $\tau_{10}:J^1P\rightarrow P$, because under the bundle identification $J^1P\simeq p^*P$ the map $\tau_{10}$ reduces to
  \[
  \tau_{10}\left(\left[j_x^1s\right]_G,u\right)=u,
  \]
  and the proposition follows.
\end{proof}

We are ready to formulate the universal property of $\omega$.

\begin{proposition}
  For every connection $\Gamma$ on $P$ we have that $\tilde\sigma_\Gamma^*\omega=\omega_\Gamma$, where $\omega_\Gamma\in\Omega^1\left(P,\g\right)$ is the connection form associated to $\Gamma$.
\end{proposition}
\begin{proof}
  For $X\in T_uP,x=\tau\left(u\right)$ we have that
  \begin{align*}
    \left(\left.\tilde\sigma_\Gamma^*\omega\right)\right|_u\left(X\right)&=\left.\omega\right|_{\tilde\sigma_\Gamma\left(u\right)}\left(\tilde\sigma_{\Gamma*}\left(X\right)\right)\\
    &=\left.\omega\right|_{\tilde\sigma_\Gamma\left(u\right)}\left(\left(T_x\sigma_\Gamma\right)\left(T_u\tau\left(X\right)\right),X\right).
  \end{align*}
  Now we have that the definition of $\omega$ involves the projection $T_{j_x^1s}\tau_{10}:T_{j_x^1s}J^1P\rightarrow T_uP$, namely
  \[
  \left.\omega\right|_{j_x^1s}\left(Z\right)=\left[j_x^1s\right]_G\left(T_{j_x^1s}\tau_{10}\left(Z\right)\right)
  \]
  for every $Z\in T_{j_x^1s}J^1P$; under the identification $J^1P\simeq p^*P$ we have that $\tau_{10}$ is the projection onto the second factor, so
  \[
  \left.\omega\right|_{\tilde\sigma_\Gamma\left(u\right)}\left(\left(T_x\sigma_\Gamma\right)\left(T_u\tau\left(X\right)\right),X\right)=\left[\tilde\sigma_\Gamma\left(u\right)\right]_G\left(X\right)=\sigma_\Gamma\left(u\right)\left(X\right)=\left.\omega_\Gamma\right|_u\left(X\right).
  \]
  Then $\tilde\sigma_\Gamma^*\omega=\omega_\Gamma$, as we want to show.
\end{proof}

\subsubsection{Local expressions}\label{subsubsect:LocalExpressions}

It is time to describe locally these constructions, in order to find expressions in the coordinates usually found in the literature; we will use the references \cite{KN1,Naka,nla.cat-vn1828701} in this task.

Let $U\subset M$ be a coordinate neighborhood and $p:LM\rightarrow M$ the canonical projection of the frame bundle; on $p^{-1}\left(U\right)$ can be defined the coordinate functions
\[
u\in p^{-1}\left(U\right)\mapsto\left(x^\mu\left(u\right),e^\nu_k\left(u\right)\right)
\]
where $x^\mu\equiv x^\mu\circ p$ and
\[
u=\left\{e_1^\mu\left.\frac{\partial}{\partial x^\mu}\right|_{p\left(u\right)},\cdots,e_n^\mu\left.\frac{\partial}{\partial x^\mu}\right|_{p\left(u\right)}\right\}.
\]
If $\bar{U}\subset M$ is another coordinate neighborhood such that $U\cap\bar{U}\not=\emptyset$ and $u\in U\cap\bar{U}$, then
\[
u=\left\{\bar{e}_1^\mu\left.\frac{\partial}{\partial \bar{x}^\mu}\right|_{p\left(u\right)},\cdots,\bar{e}_n^\mu\left.\frac{\partial}{\partial \bar{x}^\mu}\right|_{p\left(u\right)}\right\},
\]
and the coordinates change on $p^{-1}\left(U\right)\cap p^{-1}\left(\bar{U}\right)\subset LM$ can be given as
\begin{align*}
  \bar{x}^\mu&=\bar{x}^\mu\left(x^1,\cdots,x^n\right)\\
  \bar{e}_k^\mu&=\frac{\partial\bar{x}^\mu}{\partial x^\nu}e_k^\nu.
\end{align*}
On the jet space $J^1E$ of any bundle $E\rightarrow M$, the change of adapted coordinates given by the rule $\left(x^\mu,u^A\right)\mapsto\left(\bar{x}^\nu\left(x\right),\bar{u}^B\left(x,u\right)\right)$ on $E$, transform the induced coordinates on $J^1E$ accordingly to \cite{saunders89:_geomet_jet_bundl}\\
\[
\bar{u}^A_\mu=\left(\frac{\partial\bar{u}^A}{\partial u^B}u^B_\nu+\frac{\partial\bar{u}^A}{\partial{x}^\nu}\right)\frac{\partial x^\nu}{\partial\bar{x}^\mu}.
\]
By supposing that the induced coordinates on $J^1LM$ are in the present case $\left(x^\mu,e^\mu_k,e^\mu_{k\nu}\right)$ and $\left(\bar{x}^\mu,\bar{e}^\mu_k,\bar{e}^\mu_{k\nu}\right)$, we will have that
\[
\bar{e}^\mu_{k\nu}=\left(\frac{\partial\bar{x}^\mu}{\partial x^\sigma}e^\sigma_{k\rho}+\frac{\partial^2\bar{x}^\mu}{\partial{x}^\rho\partial{x}^\sigma}{e}^\sigma_k\right)\frac{\partial x^\rho}{\partial\bar{x}^\nu}.
\]
Take note on the fact that the functions
\[
\Gamma_{\mu\nu}^\sigma:=-e^\sigma_{k\nu}e^k_\mu,
\]
where the quantities $e^k_\mu$ are uniquely determined by the conditions
\[
e_\mu^ke_k^\nu=\delta^\nu_\mu,
\]
transform accordingly to
\[
\bar{\Gamma}^\mu_{\rho\gamma}=-\frac{\partial\bar{x}^\mu}{\partial x^\nu}\frac{\partial x^\sigma}{\partial\bar{x}^\alpha}e^\alpha_{k\sigma}\bar{e}^k_\rho-\frac{\partial^2\bar{x}^\mu}{\partial{x}^\rho\partial{x}^\alpha}\frac{\partial x^\alpha}{\partial\bar{x}^\gamma}.
\]
But by using the previous definition, we can find the way in which $e^k_\mu$ and $\bar{e}^k_\mu$ are related, namely
\[
\bar{e}^k_\mu=\frac{\partial x^\gamma}{\partial \bar{x}^\mu}e_\gamma^k
\]
and therefore
\[
\bar{\Gamma}^\mu_{\delta\nu}=\frac{\partial\bar{x}^\mu}{\partial x^\sigma}\frac{\partial x^\rho}{\partial\bar{x}^\nu}\frac{\partial x^\gamma}{\partial\bar{x}^\delta}\Gamma^\sigma_{\gamma\rho}-\frac{\partial^2\bar{x}^\mu}{\partial{x}^\rho\partial{x}^\gamma}\frac{\partial x^\rho}{\partial \bar{x}^\nu}\frac{\partial x^\gamma}{\partial \bar{x}^\delta},
\]
which is the transformation rule for the Christoffel symbols, if the following identity
\[
\frac{\partial^2\bar{x}^\sigma}{\partial x^\rho\partial x^\gamma}\frac{\partial x^\rho}{\partial \bar{x}^\mu}\frac{\partial x^\gamma}{\partial \bar{x}^\nu}=-\frac{\partial^2 x^\rho}{\partial\bar{x}^\mu \partial\bar{x}^\nu}\frac{\partial \bar{x}^\sigma}{\partial x^\rho}.
\]
is used. So we are ready to calculate local expressions for the previously introduced canonical forms. First we have that
\[
\theta^k=e^k_\mu\dif x^\mu
\]
determines the components of the tautological form on $J^1LM$, and the canonical connection form will result
\[
\omega^k_l=e^k_\mu\left(\dif e^\mu_l-e^\mu_{l\sigma}\dif x^\sigma\right).
\]
It is immediate to show that
\[
\bar{\theta}^k=\theta^k,
\]
and moreover
\begin{align*}
  \bar{\omega}^k_l&=\bar{e}^k_\mu\left(\dif \bar{e}^\mu_l-\bar{e}^\mu_{l\nu}\dif \bar{x}^\nu\right)\\
  &=\frac{\partial x^\gamma}{\partial\bar{x}^\mu}e^k_\gamma\left[\dif\left(\frac{\partial\bar{x}^\mu}{\partial x^\gamma}e^\gamma_l\right)-\left(\frac{\partial\bar{x}^\mu}{\partial x^\sigma}e^\sigma_{l\rho}+\frac{\partial^2\bar{x}^\mu}{\partial{x}^\rho\partial{x}^\sigma}{e}^\sigma_l\right)\frac{\partial x^\rho}{\partial\bar{x}^\nu}\dif\bar{x}^\nu\right]\\
  &=\frac{\partial x^\gamma}{\partial\bar{x}^\mu}e^k_\gamma\left(\frac{\partial\bar{x}^\mu}{\partial x^\gamma}\dif e^\gamma_l-\frac{\partial\bar{x}^\mu}{\partial x^\sigma}e^\sigma_{l\rho}\dif{x}^\rho\right)\\
  &=e^k_\gamma\left(\dif e^\gamma_l-e^\gamma_{l\rho}\dif{x}^\rho\right)\\
  &=\omega^k_l.
\end{align*}
The associated curvature form can be calculated according to the formula
\begin{align*}
  \Omega^k_l&:=\dif\omega^k_l+\omega^k_p\wedge\omega^p_l\\
  &=\dif\left[e^k_\gamma\left(\dif e^\gamma_l-e^\gamma_{l\rho}\dif{x}^\rho\right)\right]+e^k_\gamma\left(\dif e^\gamma_p-e^\gamma_{p\sigma}\dif{x}^\sigma\right)\wedge\left[e^p_\sigma\left(\dif e^\sigma_l-e^\sigma_{l\rho}\dif{x}^\rho\right)\right]\\
  &=\dif e^k_\gamma\wedge\left(\dif e^\gamma_l-e^\gamma_{l\rho}\dif{x}^\rho\right)-e^k_\gamma\dif e^\gamma_{l\rho}\wedge\dif{x}^\rho+\\
  &\qquad+e^k_\gamma e^p_\sigma\left[\dif e^\gamma_p\wedge\dif e^\sigma_l+\left(e^\gamma_{p\beta}\dif e^\sigma_l\wedge\dif x^\beta-e^\sigma_{l\beta}\dif e^\gamma_p\wedge\dif x^\beta\right)+e^\gamma_{p\beta}e^\sigma_{l\delta}\dif x^\beta\wedge\dif x^\delta\right]\\
  &=-e^\gamma_{l\rho}\dif e^k_\gamma\wedge\dif{x}^\rho-e^k_\gamma\dif e^\gamma_{l\rho}\wedge\dif{x}^\rho+\\
  &\qquad+e^k_\gamma e^p_\sigma\left[\left(e^\gamma_{p\beta}\dif e^\sigma_l\wedge\dif x^\beta-e^\sigma_{l\beta}\dif e^\gamma_p\wedge\dif x^\beta\right)+e^\gamma_{p\beta}e^\sigma_{l\delta}\dif x^\beta\wedge\dif x^\delta\right]
\end{align*}
where in the passage from the third to the fourth line it was used the identity
\[
\dif e^k_\gamma\wedge\dif e^\gamma_l+e^k_\gamma e^p_\sigma\dif e^\gamma_p\wedge\dif e^\sigma_l=0.
\]
Because of the identity
\[
e_\gamma^k\dif e^\gamma_p=-e^\gamma_p\dif e_\gamma^k
\]
we can reduce further the expression for $\Omega^k_l$
\[
\Omega^k_l =e^k_\gamma\left[-\dif e^\gamma_{l\rho}\wedge\dif{x}^\rho+e^p_\sigma\left(e^\gamma_{p\beta}\dif e^\sigma_l\wedge\dif x^\beta+e^\gamma_{p\beta}e^\sigma_{l\delta}\dif x^\beta\wedge\dif x^\delta\right)\right].
\]
Take note that
\begin{equation}\label{Eq:CurvaturaIntermedia}
e^l_\mu\Omega^k_l=e^k_\gamma\left(\dif\Gamma^\gamma_{\mu\rho}\wedge\dif{x}^\rho+\Gamma^\gamma_{\sigma\beta}\Gamma^\sigma_{\mu\delta}\dif x^\beta\wedge\dif x^\delta\right),
\end{equation}
so that if we fix a connection $\Gamma$ through its Christoffel symbols $\left(\Gamma^\mu_{\nu\sigma}\right)$ in the canonical basis $\left\{\partial/\partial x^\mu\right\}$, then we will have that $e^\gamma_k=\delta^\gamma_k$ and this formula reduces to
\[
\Omega^\mu_\nu:=e_k^\nu e^l_\mu\Omega^k_l=\dif\Gamma^\mu_{\nu\rho}\wedge\dif{x}^\rho+\Gamma^\mu_{\sigma\beta}\Gamma^\sigma_{\nu\delta}\dif x^\beta\wedge\dif x^\delta
\]
providing us with the usual formula for the connection in terms of the local coordinates.

Next we can provide a local expression for the map $\tilde\sigma_\Gamma:LM\rightarrow J^1LM$. First we realize that a connection $\Gamma$ is locally described by a map
\[
\Gamma:x^\mu\mapsto\left(x^\mu,\Gamma^\sigma_{\mu\nu}\left(x\right)\right);
\]
in these terms, the map $\tilde\sigma_\Gamma$ is given by
\[
\tilde\sigma_\Gamma:\left(x^\mu,e^k_\nu\right)\mapsto\left(x^\mu,e^k_\nu,-e^\mu_k\Gamma_{\mu\nu}^\sigma\left(x\right)\right).
\]
It is convenient to stress about an abuse of language committed here: We are indicating with the same symbol $\tilde\sigma_\Gamma$ either the map itself and its local version. Nevertheless, we obtain the following local expression for the connection form associated to $\Gamma$, namely
\[
\left(\tilde\sigma_\Gamma^*\omega\right)_l^k=e^k_\mu\left(\dif e^\mu_l+e^\sigma_l\Gamma^\mu_{\sigma\rho}\left(x\right)\dif x^\rho\right).
\]
In our approach this equation is equivalent to the so called \emph{tetrad postulate}, which relates the components \emph{of the same connection} in the two representations provided by the theory developed here: As a section $\Gamma$ of the bundle of connections, and as an equivariant map $\tilde\sigma_\Gamma:LM\rightarrow J^1LM$ such that the following diagram commutes 
\[
\begin{diagram}
  \node{LM}\arrow{e,t}{\tilde\sigma_\Gamma}\arrow{s,l}{\tau}\node{J^1LM}\arrow{s,r}{p_{GL\left(n\right)}}\\
  \node{M}\arrow{e,b}{\Gamma}\node{C\left(LM\right)}
\end{diagram}
\]
According to the previous discussion, the pullback of these forms along the section $s:x^\mu\mapsto\left(x^\mu,e_k^\nu\left(x\right)\right)$ provides us with the expression for the connection forms associated to the underlying moving frame
\[
e_k\left(x\right):=e^\nu_k\left(x\right)\frac{\partial}{\partial x^k};
\]
in fact, given another such section $\bar{s}:x^\mu\mapsto\left(x^\mu,\bar{e}_k^\nu\left(x\right)\right)$, there exists a map $g:x^\mu\mapsto\left(g^k_l\left(x\right)\right)\in GL\left(n\right)$ relating them, namely
\[
\bar{e}^\mu_k\left(x\right)=g^l_k\left(x\right)e_l^\mu\left(x\right)
\]
and so
\[
\bar{s}^*\left(\tilde\sigma_\Gamma^*\omega\right)_l^k=h^k_pg^q_ls^*\left(\tilde\sigma_\Gamma^*\omega\right)_q^p+h^k_p\dif g^p_l.
\]
It allows us to answer the concerns raised in the introduction: The Palatini Lagrangian is a global form on $J^1LM$, but this is false for its pullback along a local section. Namely, its global description needs the inclusion of information about the $1$-jet of the vielbein involved in the local representation of the connection.

\newcommand{\etalchar}[1]{$^{#1}$}

\end{document}